\documentclass{amsart}
\usepackage[utf8]{inputenc}
 \usepackage[toc,page]{appendix}
\usepackage{courier}
\usepackage{tikz-cd}
 \usepackage{mathtools}
 \usepackage{mathrsfs}
\usepackage{amsthm}
\usepackage{slashed}
\usepackage{ dsfont }
\usepackage{listings}
\usepackage{color} 
\definecolor{mygreen}{RGB}{28,172,0} 
\definecolor{mylilas}{RGB}{170,55,241}

\theoremstyle{section}
\newtheorem{condition}{Condition}

\newtheorem{definition}{Definition}
\newtheorem{theorem}{Theorem}

\newtheorem{corollary}{Corollary}[theorem]
\newtheorem{lemma}[theorem]{Lemma}
\newtheorem{proposition}{Proposition}

\newtheorem{remark}{Remark}
\usepackage[british]{babel}
\usepackage[a4paper,top=3cm,bottom=3cm,left=3cm,right=3cm,headsep=11pt]{geometry}
\usepackage{amsmath,amssymb,amsthm,mathtools} 
\usepackage{color,graphicx} 
\DeclareGraphicsExtensions{.pdf,.png,.jpg,.svg} 
\usepackage{hyperref} 
\usepackage{enumerate} 
\allowdisplaybreaks

\DeclarePairedDelimiter\abs{\lvert}{\rvert}%
\DeclarePairedDelimiter\norm{\lVert}{\rVert}%

\makeatletter
\let\oldabs\abs
\def\abs{\@ifstar{\oldabs}{\oldabs*}}
\let\oldnorm\norm
\def\norm{\@ifstar{\oldnorm}{\oldnorm*}}
\makeatother

\newcommand{\mcI}{\mathcal{I}}

\newcommand{\mcL}{\mathcal{L}}

\newcommand{\mcS}{\mathcal{S}}
\newcommand{\mcR}{\mathcal{R}}

\newcommand{\mcN}{\mathcal{N}}

\newcommand{\mcM}{\mathcal{M}}

\newcommand{\R}{\mathbb{R}}
\newcommand{\C}{\mathbb{C}}
\newcommand{\D}{\mathbb{D}}

\newcommand{\Z}{\mathbb{Z}}
\newcommand{\K}{\mathbb{K}}
\renewcommand{\P}{\mathbb{P}}
\newcommand{\E}{\mathbb{E}}
\newcommand{\rarrow}{\rightarrow}

\newcommand{\sig}{\sigma}
\newcommand{\eps}{\varepsilon}

\newcommand{\ind}{\mathds{1}}

\begin{document}
\title[Dominos and the sine-Gordon field]{Two-periodic weighted dominos and the sine-Gordon field at the free fermion point: I}
\author{Scott Mason\textsuperscript{*}}
\thanks{\textsuperscript{*}Department of Mathematics, KTH Royal Institute of Technology, scottm@kth.se}
\thanks{Supported by the grant KAW 2015.0270 from the Knut and Alice Wallenberg Foundation.}
\maketitle
\vspace{-2\baselineskip}
\begin{abstract}
In this paper we investigate the height field of a dimer model/random domino tiling on the plane at a smooth-rough (\emph{gas-liquid}) transition. We prove that the height field at this transition has two-point correlation functions which limit to those of the massless sine-Gordon field at the free fermion point, with parameters $(4\pi, z)$ where $z\in \R\setminus \{0\}$. The dimer model is on $\eps \Z^2$ and has a two-periodic weight structure with weights equal to either 1 or $a=1-C|z|\eps$, for  $0<\eps$ small (tending to zero). In order to obtain this result, we provide a direct asymptotic analysis of a double contour integral formula of the correlation kernel of the dimer model found by Fourier analysis. The limiting field interpolates between the Gaussian free field and white noise and the main result gives an explicit connection between tiling/dimer models and the law of a two-dimensional non-Gaussian field.
\end{abstract}

\section{Introduction}
The connection between the Gaussian free field and dimer models/tilings is well known, \cite{Kdom}. In particular, the Gaussian free field is seen in the continuum limit of a large class of dimer models \cite{CLR}, where it describes the fluctuations of the height function around their limit shape in the rough or \emph{liquid} phase \cite{K.O.S}. The Gaussian free field also appears in a large amount of other models, such as non-intersecting path models, random matrix theory, random graphs and other statistical mechanical models see e.g. \cite{BF}, \cite{Pet}, \cite{BG}, \cite{D}, \cite{RV}, \cite{GP}, \cite{BBNY}.  In this article we investigate the fluctuations of the continuum limit of the height function of a dimer model at a transition between the Gaussian free field and white noise. We call this transition a smooth-rough or \emph{gas-liquid} transition and note that the transition studied here is not the "rough-smooth" transition studied in \cite{J.M}.  One might expect, either since the covariance of height differences vary between zero and logarithmic or from renormalisation group heuristics \cite{Mus}, that a \emph{massive} Gaussian free field appears in the limit. However, this is not the case. To our knowledge, the main result of this paper is the first explicit connection between tiling/dimer models and the law of a two-dimensional non-Gaussian field.

In the physics literature, this type of continuum limit is known as a near-critical scaling limit \cite{Mus} (for a scaling limit in the rough phase one has an at-critical limit). The Gaussian free field \cite{S} is also known as the Euclidean \emph{bosonic} massless free field. We show that the fluctuations at the smooth-rough transition are described by a bosonized fermionic massive free field. This bosonized fermionic massive free field is known as the (massless) sine-Gordon field with parameters $(\beta,z)$, $\beta=4\pi$, $z\in \R\setminus\{0\}$, \cite{BW}. We give a review of a construction of the sine-Gordon field $SG(\beta,z)$, $\beta=4\pi$, via the Gaussian free field in section \ref{sinegordonsection}. In the context of periodic dimer models \cite{K.O.S}, the situation is that the magnetic coordinates of the dimer model are placed at the centre of the hole in the associated amoeba. The diameter of the hole is then set proportional to $\eps>0$, the lattice is rescaled by $\eps$ and we then consider the limit of the height field as $\eps\rarrow 0$. 
\begin{figure}[h]
\centering
\includegraphics[width = 0.5\textwidth]{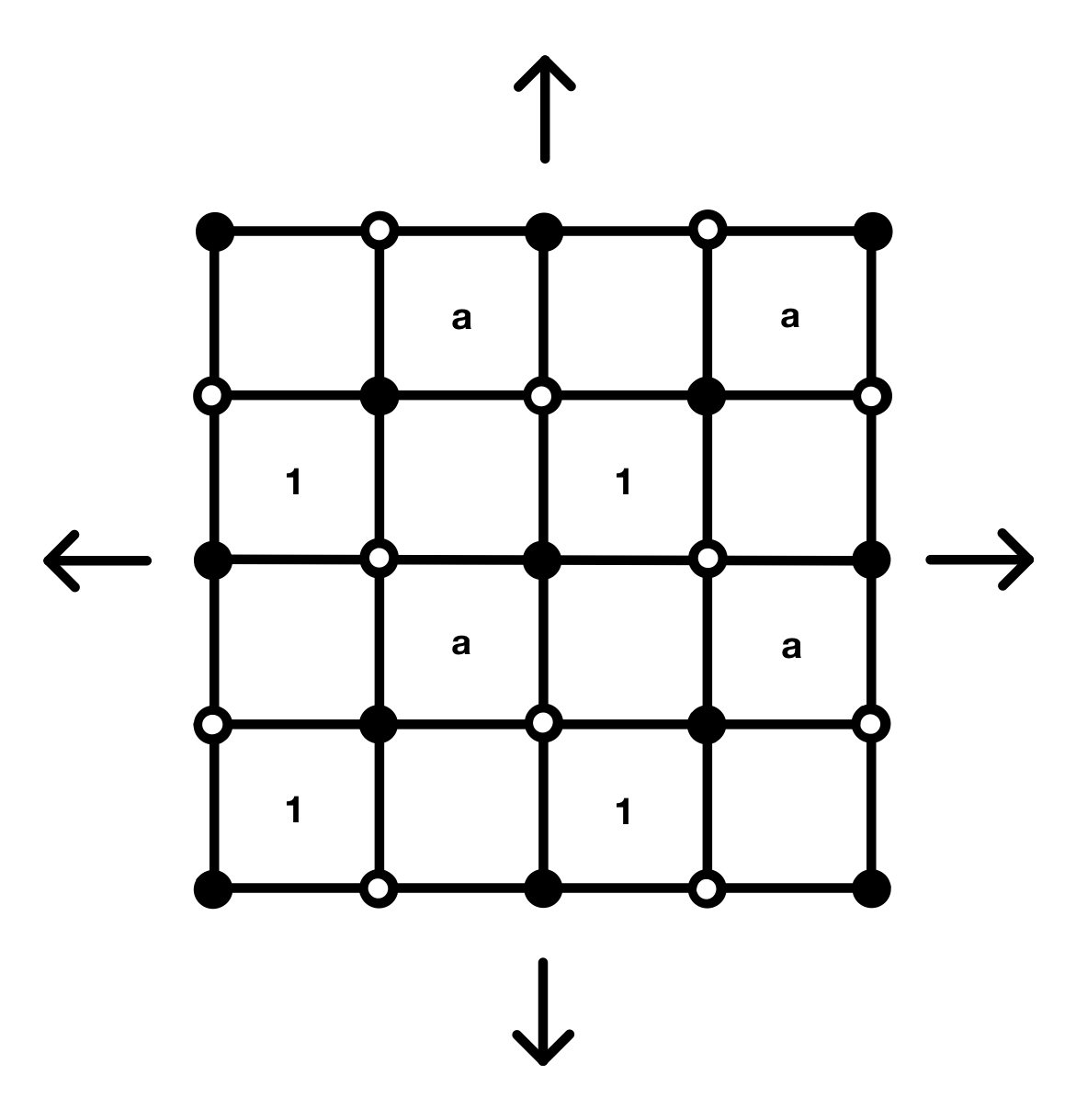}
\caption[Two-periodic weighting]
	{The two-periodic weighting on $\Z^2$. If the weight of a given face is $1$ then the edges on the boundary of that face have weight $1$, similarly if the weight of a face is $0<a<1$ then the edges on the boundary of that face have weight $a$.}
	\label{Twoperiodicweighting}
\end{figure}\\
We now give our main result, which is the convergence of the two-point correlation function of the height field to the two-point correlation function of the sine-Gordon field.
\begin{theorem}\label{mainthm}
Let $0<\eps<1$, rescale the dimer model with two-periodic weights as in figure \ref{Twoperiodicweighting} so that its vertices lie at $\eps \Z^2$. Fix $z\in \R\setminus\{0\}$ and let \begin{align}
a=1-C |z|  \eps
\end{align}
where $C =4\sqrt{2} \pi e^{-\gamma/2}$ and $\gamma$ is the Euler-Mascheroni constant.
We have the following limit,
\begin{align}
\E_a[h_\eps(f_1)h_\eps(f_2)]\rarrow \frac{1}{4\pi}\E_{SG(4\pi,z)}[\varphi(f_1)\varphi(f_2)]
\end{align}
as $\eps\rarrow 0$. \\In the above, $h_\eps$ is the smeared height function (see definition \ref{heightfielddef} in section \ref{sectionheightfield}) of the dimer model (see section \ref{dimermodelsec}), $\varphi$ is distributed under the sine-Gordon probability measure (see theorem \ref{existsineGord} in section \ref{sinegordonsection}) and $f_i\in C_c^\infty(\R^2)$ are two test functions subject to condition \ref{cond1} in section \ref{sectionheightfield}.
\end{theorem}
Theorem \ref{mainthm} is proved at the end of section \ref{sectionheightfield} where it is restated as Theorem \ref{mainthm2}. This article is structured as follows. In section \ref{dimermodelsec} we introduce the dimer model and give Theorem \ref{invkastasymp}, which regards asymptotics of its correlation kernel for $1-a$ proportional to $\eps$, when the dimers are at distance $r$ proportional to $1/\eps$.  In section \ref{sinegordonsection} we give some introductory material, followed by various formulae for the sine-Gordon field, for which we recommend \cite{BW} for further information. In section \ref{sectionheightfield} the definition of height field is given, followed by definitions, propositions and proofs related to the height field that we require to prove Theorem \ref{mainthm}. These rely on the theorem on dimer-dimer correlation asymptotics in section \ref{dimermodelsec}. In section \ref{sectionAsymptotics} we summarise a derivation of a single integral formula for the correlation kernel of the dimer model found in \cite{C/J}, starting from a double contour integral formula of the type found in \cite{K.O.S}. We then perform a rigorous asymptotic analysis of this single integral formula to prove Theorem \ref{invkastasymp}.

\begin{remark}
Similar asymptotics were computed in \cite{Ch} for the flipped/drifted dimer model (which defines the same measure as the two-periodic weighting above when e.g. the flipped model parameters are $r_1=r_4=1,r_2=r_3=a$).
 Since our inverse Kasteleyn can be related to the one found in \cite{Ch} for our weighting, the results in \cite{Ch} should agree with the asymptotics of the kernel found here, although we find small factors in the expressions which differ. We note that we have found good numerical agreement for our own asymptotics. Separately, we note that the asymptotic arguments in \cite{Ch} are somewhat unclear from a rigorous analysis perspective, here we give a direct asymptotic analysis starting from the basic formulas in \cite{K.O.S}.
 \end{remark}
 \begin{remark}
 Recently, the drifted dimer model (with parameters $s_1,s_2,s_3,s_4$) introduced in \cite{Ch} was studied in \cite{BH}. This model is equivalent to a piece of ours (with particular boundary conditions) when $s_1=s_4=1$, $s_2=s_3=a^2$, see \cite{Ch}.  Theorem \ref{mainthm} gives some information regarding a question in the open problems listed in \cite{BH}. 
 \end{remark} 
 \begin{remark} In this paper we only study the convergence of the two-point correlation function of the height field.
 The author is currently looking into obtaining full weak convergence of $h_\eps(f)$ to the sine-Gordon field.
 \end{remark}
\subsection*{Acknowledgements} I would particularly like to thank Kurt Johansson for suggesting to investigate how correlations build up as $a$ tends to 1 in the two-periodic model, for many supportive discussions along the way and for providing comments on the draft article. Many thanks goes to Sunil Chhita for helpful conversations at a conference in Stockholm. I would also like to thank both Christian Webb and Roland Bauerschmidt for discussions and helpful comments on section 3.
\section{The dimer model}\label{dimermodelsec}
In this section we introduce the relevant definitions for the dimer model, we begin with the graph and weights of the dimer model. The infinite planar graph which we are interested in is the grid graph with vertices at $\Z^2$, we denote this by $\tilde{G}=(\tilde{V}, \tilde{E})$. 
We let
\begin{align}
&\tilde{B}_0=(2\Z)\times (2\Z), && \tilde{B}_1=(2\Z+1)\times (2\Z+1),\\
&\tilde{W}_0= (2\Z+1)\times (2\Z),&& \tilde{W}_1=(2\Z)\times (2\Z+1),\nonumber
\end{align}
where $\tilde{B}=\tilde{B}_0\cup\tilde{B}_1$ are the black vertices and $\tilde{W}=\tilde{W}_0\cup\tilde{W}_1$ are the white vertices, $\tilde{V}=\tilde{B}\cup\tilde{W}$. We now define a weighting on the edge set $\tilde{E}$. First we define the "$a$-faces" of $\tilde{G}$ to be the faces with centres at the points $(2\Z+1/2)\times(2\Z+1/2)$. Define the weights of the edges on the boundary of an $a$-face to have weight $0<a<1$. Define the weights of all other edges to be equal to one. The Gibbs measure, $\P_a$, on the set of dimer configurations $\mcM(\tilde{G})$ of $\tilde{G}$ falls under a general construction given in \cite{K.O.S} and determines an infinite translation invariant gradient Gibbs measure $\mu_a$ on the set of height functions (defined below), for which $\E_a$ denotes the expectation with respect to.

\begin{remark}
In order to connect with previous formulae related to the two-periodic weighting, we rotate $\tilde{G}$ by 45 degrees clockwise around the point $(1,0)$ and expand it by factor of $\sqrt{2}$, we call this new graph $G$. If we view the edges and vertices of $G, \tilde{G}$ as subsets of $\R^2$ they are related by
\begin{align}
\tilde{G}=\frac{1}{2}\begin{pmatrix} 1 &-1\\ 1 & 1\end{pmatrix}\big(G-\begin{pmatrix} 1 \\ 0\end{pmatrix}\big)+\begin{pmatrix}1 \\ 0\end{pmatrix}.\label{changeofcoordsG}
\end{align}
We use \eqref{changeofcoordsG} to write the image of $\tilde{B}_i$ as $B_i$ and $\tilde{W}_i$ as $W_i$, $i=0,1$ and let it preserve edge weights.
\end{remark}

Explicitly for the graph $G=(V,E)$ in we have
 for $i\in\{0,1\}$,
\begin{align}
B_i=\{(x,y)\in \Z^2: x \text{ mod }2=0, y \text{ mod }2 =1, x+y \text{ mod } 4=2i+1\},\nonumber\\
W_i=\{(x,y)\in \Z^2: x \text{ mod }2=1, y \text{ mod }2 =0, x+y \text{ mod } 4=2i+1\}\nonumber
\end{align}
where $B=B_0\cup B_1$ are the black vertices, $W=W_0\cup W_1$ are the white vertices and $V=W\cup B$. The edge set $E$ is all edges of the form $b-w=\pm \vec{e}_1,\pm\vec{e}_2$, for $b\in B,w\in W$, where $\vec{e}_1=(1,1), \vec{e}_2=(-1,1)$. The weights maintain their two periodic fashion, in particular the edges along the boundary of the face centred at the point $(i,j)$ where $(i+j)$ mod $4=2$, have weight $0<a<1$ and the edges on the boundary of the face centred at the the point $(i,j)$ where $(i+j)$ mod $4=0$ have weight $1$.

This model is seen in the smooth phase of the two periodic Aztec-diamond, see \cite{C/J}. Asymptotics of its correlation kernel (defined below in section \ref{dimermodelsec}) were computed in \cite{J.M} for fixed $a\in(0,1)$. For a more complete introduction to the following see section 1.5 of \cite{J.M}. 
Define the Kasteleyn matrix
\begin{align}
\K_a(y,x)=\begin{cases} a(1-j)+j & \text{if } x=y+\vec{e}_1, y\in B_{j}\\
i(aj+(1-j)) & \text{if } x=y+\vec{e}_2, y\in B_{j}\\
aj+(1-j) & \text{if } x=y-\vec{e}_1, y\in B_{j}\\
i(a(1-j)+j) & \text{if }  x=y-\vec{e}_2,y\in B_{j}\\
0 & \text{otherwise}
\end{cases}
\end{align}
 where $i=\sqrt{-1},$ $j\in\{0,1\} $.
Suppose that $x\in W_{\eps_1}$ and $y\in B_{\eps_2}$, $\eps_1,\eps_2\in\{0,1\}$. Let $(u,v)\in \Z^2$ be such that $u (2\vec{e}_1)+v(2\vec{e}_2)$ is the translation to get from the fundamental domain containing $x$ to the fundamental domain containing $y$. The inverse Kasteleyn matrix of $G$ for the entries $x,y$ is
\begin{align}
\K_{a}^{-1}(x,y)=\frac{1}{(2\pi i)^2}\int_{\Gamma_{1}}\frac{dz}{z}\int_{\Gamma_{1}}\frac{dw}{w}\frac{Q(z,w)_{\eps_1+1,\eps_2+1}}{P(z,w)}z^uw^v
\label{infinvKasteleyn}
\end{align}
where $\Gamma_r$ is a circle or radius $r$ around the origin,
\begin{align}
Q(z,w)=\begin{pmatrix}
i(a+w) & -(a+z)\\
-(a+1/z) & i(a+1/w)
\end{pmatrix}
\end{align}
and 
\begin{align}
P(z,w)=-2-2a^2-\frac{a}{w}-aw-\frac{a}{z}-az
\end{align}
is the characteristic polynomial.
Write $e_i=(b_i,w_i)\in B_{\eps_2}\times W_{\eps_1}$ for $1\leq i \leq n$. The probabilities of cylinder sets, i.e. sets of the form $\{e_1,...,e_n\in \omega\}\subset\mcM(G)$, are given by 
\begin{align}\label{dimersdetprocess}
\P_{a}(e_1,...,e_n\in \omega)=\det(L(e_i,e_j))_{i,j}^n.
\end{align}
where the \emph{correlation kernel} is given by
\begin{align}
L(e_i,e_j)=\K_{a}(b_i,w_i)\K_{a}^{-1}(w_j,b_i).
\label{corrkernel}
\end{align}
Note that the formula in \eqref{dimersdetprocess} defines the Gibbs measure, $\P_a$.

The essential asymptotic ingredient we use to prove Theorem \ref{mainthm} is an asymptotic expansion of $L(e_i,e_j)$ when the dimers $e_i,e_j$ are  in distance $ r/\eps$ apart whilst simultaneously taking the parameter limit $a=1-\lambda \eps$, for $\lambda>0$ fixed, $r>0$ in a compact set, $\eps>0$ tending to zero. In this context we have the following theorem.
Denote the vectors
\begin{align}
\vec{e}_1=(1,1),&& \vec{e}_2=(-1,1).
\end{align}
\begin{theorem}\label{invkastasymp}
Let $a=1-\lambda\eps>0$, $\eps>0$ and $\lambda>0$ be fixed independent of $\eps$. Index the elements of $\K_a^{-1}$ as $(x(j),y(i))\in W_{\eps_1}\times B_{\eps_2}$ with coordinates given by
\begin{align}
x(j)=(x_1(j),x_2(j))=(1+\alpha_j\eps^{-1})\vec{e}_1+\beta_j\eps^{-1}\vec{e}_2+(0,2\eps_1-1)\label{temp039112}\\
y(i)=(y_1(i),y_2(i))=(1+\alpha_i\eps^{-1})\vec{e}_1+\beta_i\eps^{-1}\vec{e}_2+(2\eps_2-1,0)\nonumber
\end{align}
where $i,j\in\{0,1\}$. In \eqref{temp039112} we take $(\alpha_k\eps^{-1},\beta_k\eps^{-1})\in (2\Z)^2$, $k\in\{0,1\}$  and $(\alpha_0,\beta_0)\neq (\alpha_1,\beta_1)$.\\
Let 
\begin{align}
\alpha=\alpha_j-\alpha_i, && \beta=\beta_j-\beta_i,
\end{align}
where $\{0,1\}\ni i\neq j$, and set $z=\sqrt{\alpha^2+\beta^2}$. 
We have the following limits; \\
if $\eps_1=0,\eps_2=0$ then
\begin{align}
\frac{(-1)^{(\alpha+\beta)\eps^{-1}/2}}{\eps}\K_a^{-1}(x(j),y(i))=\frac{i\lambda}{2\pi}K_0(\lambda z/\sqrt{2})+\frac{i\lambda\beta}{\sqrt{2}\pi z}K_1(\lambda z/\sqrt{2})+o(1),
\end{align}
if $\eps_1=0,\eps_2=1$ then
\begin{align}
\frac{(-1)^{(\alpha+\beta)\eps^{-1}/2}}{\eps}\K_a^{-1}(x(j),y(i))=-\frac{\lambda}{2\pi}K_0(\lambda z/\sqrt{2})-\frac{\lambda\alpha}{\sqrt{2}\pi z}K_1(\lambda z/\sqrt{2})+o(1),
\end{align}
if $\eps_1=1,\eps_2=0$ then
\begin{align}
\frac{(-1)^{(\alpha+\beta)\eps^{-1}/2}}{\eps}\K_a^{-1}(x(j),y(i))=-\frac{\lambda}{2\pi}K_0(\lambda z/\sqrt{2})+\frac{\lambda\alpha}{\sqrt{2}\pi z}K_1(\lambda z/\sqrt{2})+o(1),
\end{align}
if $\eps_1=1,\eps_2=1$ then 
\begin{align}
\frac{(-1)^{(\alpha+\beta)\eps^{-1}/2}}{\eps}\K_a^{-1}(x(j),y(i))=\frac{i\lambda}{2\pi}K_0(\lambda z/\sqrt{2})-\frac{i\lambda\beta}{\sqrt{2}\pi z}K_1(\lambda z/\sqrt{2})+o(1)
\end{align}
uniformly for $(\alpha_0,\beta_0), (\alpha_1,\beta_1)$ in two compact, disjoint subsets of $\R^2$ as $\eps>0$ tends to zero. Here $K_\nu$ are the modified Bessel functions of the second kind (see the Appendix).
\begin{proof}
see the end of section \ref{proofinvKastasymp}.
\end{proof}
\end{theorem}

\section{The sine-Gordon field}\label{sinegordonsection}
In this section we give some introductory material on the sine-Gordon field, see \cite{BW} for more details and references therein. 
Informally, one can think of the sine-Gordon field with parameters $(\beta,z)$, $\beta\in (0,8\pi]$, $z\in\R$, as something like a probability measure obtained by reweighting the Gaussian free field $\mu_{GFF}$ as
\begin{align}
\frac{1}{\mathcal{Z}}\exp\big[{2z\int dx :\cos(\sqrt{\beta}\varphi(x)): \  }\big]\mu_{GFF}(d\varphi).
\end{align}
This expression does not make sense as it is currently written. For example, the Gaussian free field is a measure on linear functionals, so the pointwise evaluation $\varphi(x)$ does not make sense. In the following we define the sine-Gordon field more precisely. If the reader is already familiar with the Gaussian free field and regularisation, they may wish to skip to theorem \ref{existsineGord}.\\

For pedagogical reasons, we will define the massive/massless GFF as just a probability measure on two particular duals of "nuclear" spaces (we will not define nuclear spaces in general here). We also introduce a regularised GFF, which is necessary to understand the definition of the sine-Gordon field \cite{BW}. We draw from \cite{GV}, but recommend \cite{BDW}, \cite{LSSW} and \cite{GF}.\\
Denote the Schwartz space $\mcS(\R^2)$ to be the set of smooth real-valued functions on $\R^2$ whose derivatives all decay faster than any polynomial at infinity. Denote the subspace of functions in $\mcS(\R^2)$ with mean zero as $\mcS_0(\R^2)$, that is
\begin{align}
\mcS_0(\R^2)=\big\{f\in\mcS(\R^2) \ ; \ \int_{\R^2}f(x)dx=0 \big\}.
\end{align} We define the topology on $\mcS(\R^2)$ to be generated by the family of semi-norms 
\begin{align}
\big\{||f||_{n,i,j}:=\sup_{x\in \R^2}|x|^n|\partial^i_{x_1}\partial^j_{x_2}f(x) \ ; \ n,i,j\geq 0\big\},
\end{align}
which induces the subspace topology on $\mcS_0(\R^2)$.
The space of tempered distributions $\mcS'(\R^2)$ is defined as the space of continuous linear functionals on $\mcS(\R^2)$. Similarly, the space of tempered distributions modulo additive constants $\mcS_0'(\R^2)$ is defined as the space of continuous linear functionals on $\mcS_0$. Observe that $\mcS'(\R^2)\subset \mcS_0'(\R^2)$ because $\mcS_0(\R^2)\subset \mcS(\R^2)$.\\
Let $X$ denote either  $\mcS(\R^2)$ or $\mcS_0(\R^2)$, and $X'$ denote the topological dual of $X$. 
We introduce a class of measures on $X'$. A \emph{cylinder measure on $X'$} is a measure defined on the $\sigma$-algebra generated by all sets of the form
\begin{align}
\{F\in X'\ ; \ F(f)\in B\}
\end{align}
where $f\in X$ and $B$ is Borel in $\R$.
The Bochner-Minlos theorem gives necessary and sufficient conditions for a cylinder probability measure $\mu$ to exist on the dual of a \emph{nuclear space}, see chapter 4 in \cite{GV}. However, since we are only interested in the two nuclear spaces $\mcS$ and $\mcS_0$ we restate the theorem in our context. These conditions are in terms of properties of the \emph{characteristic functional}, which  is defined by $\mcL:X\rarrow \C$;
\begin{align}
\mcL(f)=\int_{X'} e^{iF(f)}d\mu(F).
\end{align}
\begin{theorem}[Bochner-Minlos for either $\mcS'(\R^2)$ or $\mcS_0'(\R^2)$]
Let $X$ denote either  $\mcS(\R^2)$ or $\mcS_0(\R^2)$, and $X'$ denote the topological dual of $X$. A function $\mcL:X\rarrow \C$ is the characteristic functional of a probability measure $\mu$ on $X'$ if and only if $\mcL(0)=1$, $\mcL$ is continuous and $\mcL$ is positive definite, which means
\begin{align}
\sum_{j,k=1}^nz_j \overline{z}_k\mcL(f_j-f_k)\geq 0
\end{align}
for all $f_1,...,f_n\in X$ and $z_1,...,z_n\in \C$.
\end{theorem}
The Gaussian free field with "mass" $m>0$ is the probability measure on $\mcS'(\R^2)$ defined by the characteristic functional
\begin{align}
\mcL_{GFF(m)}(f)=e^{-\frac{1}{2}(f,(-\nabla^2+m^2)^{-1} f)_{L_2}}\label{temp4204}
\end{align}
for $f\in \mcS(\R^2)$. The case $m=0$ is defined below for comparison. The operator $(-\nabla^2+m^2)^{-1}$ is an integral operator with an integral kernel given by 
\begin{align}
\frac{1}{2\pi}K_0(m|x-y|)=\int_0^\infty ds\frac{e^{-|x-y|^2/4s-m^2s}}{4\pi s} \label{K0intrep}
\end{align}
where $K_0$ is a modified Bessel function of the second kind and the integral representation \eqref{K0intrep} can be found in \cite{DLMF}.
By changing $f$ to $tf$ in \eqref{temp4204} for a parameter $t\in \R$, observe we have now specified the Fourier transform of a collection of mean zero Gaussian random variables on $\R$ indexed by $\mcS(\R^2)$, that is for each fixed $f\in\mcS(\R^2)$:
\begin{align}
F(f)\sim \mcN(0,(f,(-\nabla^2+m^2)^{-1} f)_{L_2}).
\end{align}
The covariances between these random variables are given by ($f,g\in\mcS(\R^2)$);
\begin{align}
&\text{Cov}[F(f)F(g)]=(f,(-\nabla^2+m^2)^{-1} g)_{L_2}\\
&=\int_{\R^2\times \R^2} f(x)\frac{1}{2\pi}K_0(m|x-y|)g(y)dxdy\nonumber\\&=\int_{\R^2\times \R^2} f(x)\int_0^\infty ds\frac{e^{-|x-y|^2/4s-m^2s}}{4\pi s}g(y)dxdy.\nonumber
\end{align}

Recall that the inverse Laplace operator $(-\nabla^{2})^{-1}$ is an integral operator with integral kernel given by $-\frac{1}{2\pi}\log(|x-y|)$.
The Gaussian free field with mass $m=0$ is not defined by the characteristic functional \eqref{temp4204}, and $(f,(-\nabla^2)^{-1}f)_{L_2}$ is not even finite for all $f\in \mcS(\R^2)$. If we denote $\hat{f}$ as the Fourier transform of $f\in S_0$ then indeed,
\begin{align}
(f,(-\nabla^{2})^{-1}f)_{L_2}&=\int_{\R^2\times \R^2}f(x)(-\frac{1}{2\pi}\log|x-y|)f(y)dxdy\label{temp02586}\\
&=\frac{1}{2\pi}\int_{\R^2} |\omega|^{-2}|\hat{f}(\omega)|^2d\omega\nonumber
\end{align}
is finite since $\hat{f}(0)=0$ iff $\int f(x)dx=0$. One can show that \eqref{temp02586} is not integrable when $f\in \mcS\setminus \mcS_0$. However, we can restrict the characteristic functional to $f\in\mcS_0$ and define the massless Gaussian free field $\mu_{GFF(0)}$ on $\mcS_0'(\R^2)$ precisely via
\begin{align}
\mcL_{GFF(0)}(f)=e^{-\frac{1}{2}(f,(-\nabla^2)^{-1} f)_{L_2}}.
\end{align}

We now introduce a regularisation of  the \emph{massive} Gaussian free field $\mu_{GFF(m)}$, $m>0$, which appears in our regularisation of the sine-Gordon field \eqref{temp3023}.
If one replaces the integral operator $(-\nabla^2+m^2)^{-1}$ in the characteristic functional \eqref{temp4204} 
 by the integral operator with integral kernel
 \begin{align}
 \int_{\eps^2}^\infty ds\frac{e^{-|x-y|^2/4s-m^2s}}{4\pi s},\quad \eps>0,
 \end{align}  then one obtains a new characteristic functional $\mcL_{GFF(m,\eps)}$. The associated Gaussian measure $\mu_{GFF(m,\eps)}$ of $\mcL_{GFF(m,\eps)}$ localises on smooth tempered distributions \cite{BW} - that is - $\mu_{GFF(m,\eps)}$ has its full support given by the subspace of linear functionals in $\mcS'$ of the form 
 \begin{align}
 F(f)=\int \psi(x)f(x)dx,\quad \psi\in C^\infty.\label{temp1094}
 \end{align}
  We can consider  \eqref{temp1094}  as a map from (a subset of) $C^\infty$ into $\mcS'$. This map is an injection and allows us to define a pointwise evaluation of a smooth tempered distribution $F$ via
  \begin{align}
  F(x):= \lim_{\delta\rarrow 0}\int \psi(y)\frac{e^{-|x-y|^2/(2\delta)}}{2\pi\delta^2}dy=\int \psi(y)\delta(y-x)dy=\psi(x).
  \end{align}
We can instead consider $\mu_{GFF(m,\eps)}$ as a measure directly on $C^\infty$, we call this new probability measure  a \emph{regularisation} of $\mu_{GFF(m)}$.

We are now ready to introduce the sine-Gordon field. One would like to define the \emph{massless} sine-Gordon field $SG(\beta,z)$ with $\beta\in (0,8\pi)$, $z\in \R$ as the weak limit $L\rarrow \infty, m\rarrow 0, \eps\rarrow 0$ of the probability measure specified by
\begin{align}
d\mu_{SG(\beta,z \ : \ \eps,m,L)}(\varphi)\propto \exp \Big [2z\int_{\Lambda} \eps^{-\beta/4\pi}\cos(\sqrt{\beta}\varphi(x))dx\Big]d\mu_{GFF(\eps,m)}(\varphi)\label{temp3023}
\end{align}
where $\Lambda:=\Lambda_L=\{x\in \R^2  \ ; \ |x|\leq L\}$. Note that due to the regularisation $\mu_{GFF(m,\eps)}$ of $\mu_{GFF(m)}$, the pointwise evaluation $\varphi(x)$ in \eqref{temp3023} makes sense. 

In this context, let $\mcL_{SG(\beta,z \ :  \  \eps,m,L)}(f)$ denote the characteristic functional of \eqref{temp3023}.
It was proven in \cite{BW} (see Theorem 1.6 and its proof) that for $\beta\in(0,6\pi)$, $z\in \R$, the functional defined by the limit
\begin{align}
\lim_{L\rarrow \infty}\lim_{m\rarrow 0}\lim_{\eps\rarrow 0}\mcL_{SG(\beta,z \ :  \  \eps,m,L)}(f)
\end{align}
for $f\in C_c^\infty(\R^2)$ with $\int f=0$, is uniformly continuous in the topology of $\mcS_0$ and extends to a unique characteristic functional $\mcL_{SG(\beta,z)}^*(f)$ on $\mcS_0$ of a probability measure $\mu_{SG(\beta,z)}^*$. 

Let $\gamma_N(x)=e^{-|x|^2/(2N)}/(2\pi N)$ denote a Gaussian density on $\R^2$ with mean zero and variance $N$.
In the case $\beta=4\pi$, we have the following further result (see Theorem 1.3 and its proof in \cite{BW}):
\begin{theorem}[Existence of $\varphi$ \cite{BW}]\label{existsineGord}
Let $z\in\R\setminus \{0\}$, define an extension of $\mcL_{SG(4\pi,z)}^*$ to $f\in\mcS$ by
\begin{align}
\mcL_{SG(4\pi,z)}(f):=\lim_{N\rarrow \infty}\mcL_{SG(4\pi,z)}^*(f-\hat{f}(0)\gamma_N).
\end{align} 
Then $\mcL_{SG(4\pi,z)}$ defines the characteristic functional of a probability measure $\mu_{SG(4\pi,z)}$  on $\mcS'$ (not $\mcS_0'$) and, in this paper, we call $\mu_{SG(4\pi,z)}$ the probability measure of the sine-Gordon field $SG(4\pi,z)$. We denote the expectation with respect to $\mu_{SG(4\pi,z)}$ by $\E_{SG(4\pi,z)}$.
\end{theorem}
\begin{remark}
It is natural to consider the limiting functional 
\begin{align}
\lim_{m\rarrow 0}\lim_{L\rarrow \infty}\lim_{\eps\rarrow 0}\mcL_{SG(\beta ,z \ :  \  \eps,m,L)}(f)
\end{align} for $f\in C_c^\infty(\R^2)$ (not assuming $\int f=0$)  for a construction of the sine-Gordon field $SG(\beta ,z)$. Indeed, in the case $\beta =4\pi$, \cite{BW} show that this limiting functional has a unique extension to $\mcS$, and defines a probability measure on $\mcS'$. It is currently an open problem to determine whether this extension is equivalent to $\mcL_{SG(4\pi,z)}$ or not, however we suspect it is.
\end{remark}

The paper \cite{BW} also give formulas for the covariances of the sine-Gordon field for $(4\pi,z)$, $z\neq 0$.
For the two point correlation function of $SG(4\pi,z)$,
\begin{align}
\E_{SG(4\pi,z)}[\varphi(f_1)\varphi(f_2)]=\int_{\R^2}\frac{dp}{(2\pi)^3}\widehat{f_1}(p)\widehat{f_2}(-p)C_{A|z|}(p)
\end{align}
where $\widehat{f}(p)=\int f(x)e^{-ip\cdot x}dx$ and
\begin{align}
C_{\mu}(p)=\mu^{-2}F(|p|/\mu), \quad \text{with } \quad F(x)=\frac{1}{x^2}-4\frac{\text{arcsinh}(x/2)}{x^3\sqrt{4+x^2}}.
\end{align}
The same paper also gives formulae for the covariances of the field's weak derivatives (in a sense defined below). In order to define this we introduce some notation following \cite{BW}.
Identify points in the real plane $(x_0,x_1)\in \R^2$ with points in the complex plane $ix_0+x_1\in\C$. Define the differential operators 
\begin{align}
\partial =\frac{1}{2}(-i\partial_0+\partial_1), && \overline{\partial}=\frac{1}{2}(i\partial_0+\partial_1),\label{diffops}
\end{align}
 where $\partial_j$ is the directional derivative along $x_j$, $j=0,1$.

For smooth, compactly supported functions $f_1,f_2\in C_c^\infty(\R^2)$, define
\begin{align}
\E_{SG(4\pi,z)}\big [\partial \varphi(f_1)\partial \varphi(f_2)\big]=\lim_{L\rarrow \infty}\lim_{m\rarrow 0}\lim_{\eps\rarrow 0}\E_{SG(4\pi,z \ : \ \eps,m,L)}\big [\partial \varphi(f_1)\partial \varphi(f_2)\big],\label{temp014}
\end{align}
similarly for $\E_{SG(4\pi,z)}\big [\partial \varphi(f_1)\overline{\partial} \varphi(f_2)\big]$. Also observe that via integration by parts, $\partial\varphi(f)=-\varphi(\partial f)$ for smooth functions $\varphi,\partial\varphi$ considered as linear functionals on $C_c^\infty$. 
\begin{theorem}[\cite{BW}]\label{theorem2BW}
Let $\beta=4\pi$ and $z\in \R\setminus \{0\}$. Then for $f_1,f_2\in C_c^\infty (\R^2)$ with disjoint support,
\begin{align}
&\E_{SG(4\pi,z)}\big [\partial \varphi(f_1)\partial \varphi(f_2)\big]=-\frac{B^2}{\pi^2}p.v\int dx_1dx_2f_1(x_1)f_2(x_2)(\partial_{x_1}K_0(A|z||x_1-x_2|))^2,\\
&\E_{SG(4\pi,z)}\big [\partial \varphi(f_1)\overline{\partial} \varphi(f_2)\big]=-\frac{B^2A^2z^2}{4\pi^2}\int dx_1dx_2f_1(x_1)f_2(x_2)(K_0(A|z||x_1-x_2|))^2\label{temp34fr}
\end{align}
where $p.v\int$ stands for the integral $\lim_{\delta\rarrow 0}\int_{|x_1-x_2|\geq \delta}$, $A=4\pi e^{-\gamma/2}$, $B=\sqrt{\pi}$ and where $\gamma$ is the Euler-Mascheroni constant and $\partial_x$ means $\partial$ applied in the $x$-variable.
\end{theorem}
It is straightforward to use the linearity of expectation values and limits in \eqref{temp014} to obtain the following lemma.
\begin{lemma}[]\label{dirderivSG}
Let $\beta=4\pi$, $z\in \R\setminus \{0\}$  and denote the variables $x=(x_0,x_1)$, $y=(y_0,y_1)$. Then for $f_1,f_2\in C_c^\infty (\R^2)$ with disjoint support, we have the formulas
\begin{align}
\E_{SG(4\pi,z)}\big [\partial_0 \varphi(f_1)\partial_1 \varphi(f_2)\big]=& \ -\frac{B^2A^2z^2}{\pi^2}\int dxdyf_1(x)f_2(y)\frac{(x_0-y_0)(x_1-y_1)}{|x-y|^2}K_1(A|z||x-y|)^2,\\
\E_{SG(4\pi,z)}\big [\partial_1 \varphi(f_1)\partial_0 \varphi(f_2)\big]=& \ -\frac{B^2A^2z^2}{\pi^2}\int dxdyf_1(x)f_2(y)\frac{(x_0-y_0)(x_1-y_1)}{|x-y|^2}K_1(A|z||x-y|)^2,\\
\E_{SG(4\pi,z)}\big [\partial_0 \varphi(f_1)\partial_0 \varphi(f_2)\big]=& \ \frac{B^2A^2z^2}{2\pi^2}\int dxdyf_1(x)f_2(y)\big [-K_0(A|z||x-y|)^2\\&\qquad\qquad+\Big (-\Big(\frac{x_0-y_0}{|x-y|}\Big)^2+\Big(\frac{x_1-y_1}{|x-y|}\Big)^2\Big)K_1(A|z||x-y|)^2\big],\nonumber\\
\E_{SG(4\pi,z)}\big [\partial_1 \varphi(f_1)\partial_1 \varphi(f_2)\big]=& \ \frac{B^2A^2z^2}{2\pi^2}\int dxdyf_1(x)f_2(y)\big [-K_0(A|z||x-y|)^2\\&\qquad\qquad+\Big (\Big(\frac{x_0-y_0}{|x-y|}\Big)^2-\Big(\frac{x_1-y_1}{|x-y|}\Big)^2\Big)K_1(A|z||x-y|)^2\big]\nonumber
\end{align}
where $A=4\pi e^{-\gamma/2}$, $B=\sqrt{\pi}$ and $\gamma$ is the Euler-Mascheroni constant.
\end{lemma}

 Define the Dirac operator with mass $\mu$ by
\begin{align} 
i\slashed{\partial}+\mu=
\begin{pmatrix}
\mu & 2\overline{\partial}\\
2\partial & \mu
\end{pmatrix}.
\end{align}
The inverse of the Dirac operator has integral kernel given by 
\begin{align}
S(x,y)=-\frac{1}{2\pi}\begin{pmatrix}-\mu K_0(|\mu||x-y|)& 2\overline{\partial}_xK_0(|\mu||x-y|)\\
2\partial_xK_0(|\mu||x-y|)& -\mu K_0(|\mu||x-y|).
\end{pmatrix}
\end{align}
For distinct points $x_1,...,x_n,y_1...,y_n\in \R^2$ and $\alpha_1...,\alpha_n,\beta_1,...,\beta_n\in\{1,2\}$, the correlation functions of free Dirac fermions are given by
\begin{align}
\Big \langle\prod_{i=1}^n\overline{\psi}_{\alpha_i}(x_i)\psi_{\beta_i}(y_i)\Big \rangle_{FF(\mu)}=\det (S_{\alpha_i,\beta_j}(x_i,y_j))_{i,j=1}^n.
\end{align}
This is singular when $x_j=y_j$ for some $j$. However one may consider the \emph{truncated} correlation functions (also known as cumulants) which is not singular in this case. For more on this see \cite{BW}, we merely note that for our purposes the truncated correlation functions for $n=2$ are given by
\begin{align}
\Big \langle\overline{\psi}_{\alpha_1}\psi_{\beta_1}(x_1)\overline{\psi}_{\alpha_2}\psi_{\beta_2}(x_2)\Big \rangle_{FF(\mu)}^T=-S_{\alpha_2,\beta_1}(x_2,x_1)S_{\alpha_1,\beta_2}(x_1,x_2)
\end{align}
where $x_1,x_2\in\R^2$ are distinct. In fact, one can get the above formula involving $\overline{\psi}_{\alpha}\psi_{\beta}(x)$ under a limit of a regularisation of $\overline{\psi}_{\alpha}(x)\psi_{\beta}(x)$.  Also define the two-point truncated correlation functions pairing with test functions as
\begin{align}
\Big \langle\overline{\psi}_{\alpha_1}\psi_{\beta_1}(f_1)\overline{\psi}_{\alpha_2}\psi_{\beta_2}(f_2)\Big \rangle_{FF(\mu)}^T=\int dx_1dx_2 f_1(x_1)f_2(x_2)\Big \langle\overline{\psi}_{\alpha_1}\psi_{\beta_1}(x_1)\overline{\psi}_{\alpha_2}\psi_{\beta_2}(x_2)\Big \rangle_{FF(\mu)}^T.\label{fermcorrpaired}
\end{align}
\begin{remark}
Theorem \ref{theorem2BW} differs from Theorem (1.2) in \cite{BW} by an overall minus sign in \eqref{temp34fr}. This is just a minor mistake. To see where it appears, in their own notation we have
\begin{align}
\langle \overline{\psi}_2\psi_1(x)\overline{\psi}_1\psi_2(y)\rangle_{FF(\mu)}^T=-S_{1,1}(y,x)S_{2,2}(x,y)=\frac{\mu^2}{4\pi^2}K_0(\mu|x-y|)^2.
\end{align} 
Using this in the Coleman correspondence (their Theorem (1.1)) with $n=n'=0$, $q'=q=1$, one sees an overall factor of $i^2=-1$ missing from their formula (1.15). 
\end{remark}

\section{The height field}\label{sectionheightfield}
In this section we define the height function of the dimer model and its random height field.

For our dimer model the corresponding height function is defined on faces $f$ of the graph $\tilde{G}$ (with vertices at $\Z^2$) and the height function $h(f)$ is specified by fixing its value at an arbitrary fixed face and then defining the following height changes between adjacent faces;\\\\
a height change of $+ 3/4$ $ (-3/4)$ when traversing across an edge covered by a dimer with the white vertex on the right (left),\\\\
a height change of $+ 1/4$ $ (-1/4)$ when traversing across an edge \emph{not} covered by a dimer with the white vertex on the left (right).\\\\
In this article we assign the height at the $a$-face $(1/2,1/2)$ to be zero, this results in the average $a$-height (see definition \ref{aheightdef}) to be zero. In the following we will often identify faces of $\tilde{G}$ with their centre points $(\Z+1/2)\times (\Z+1/2)$ so that we can consider the height function as a random map from $(\Z+1/2)\times (\Z+1/2)$ into $\Z/4$.
\begin{definition}\label{heightfielddef}
Given a smooth compactly supported function $f:\R^2\mapsto f(x)\in \R$, define for $0<\eps<1$, 
\begin{align}
h_\eps(f):=\eps^2\sum h(x/\eps)f(x)
\end{align}
where the sum is over faces $x$ with centre points $\eps(\Z+1/2)\times \eps(\Z+1/2)$. We call $h_\eps$ the height field of the height function.
\end{definition}
\begin{definition}\label{aheightdef}
Let $h^a$ denote the height function $h$ restricted to the \emph{a}-faces of $\tilde{G}$. We call $h^a$ the $a$-height function, and observe that it takes values in $\Z$.
\end{definition}

We can consider the $a$-height function as a random map $h^a:(2\Z+1/2)\times(2\Z+1/2)\rarrow \Z$. 
Define
\begin{align}
\partial_0h^a(x)=\frac{h^a(x+(2,0))-h^a(x)}{2\eps}, \quad \partial_1h^a(x)=\frac{h^a(x+(0,2))-h^a(x)}{2\eps},
\end{align}
for $x\in(2\Z+1/2)\times(2\Z+1/2)$ and where $\partial_j:=\partial_j^\eps$ depends on $0<\eps\leq 1$. 
Now we consider the $a$-height function as a random linear functional on $C_c^\infty(\R^2)$.  
\begin{definition}
Given a smooth compactly supported function $f:\R^2\mapsto f(x)\in \R$, define for $0<\eps<1$, 
\begin{align}
h_\eps^a(f):=(2\eps)^2\sum h^a(x/\eps)f(x)
\end{align}
where the sum is over $x$ in $\eps(2\Z+1/2)\times \eps(2\Z+1/2)$. We call $h_\eps^a$ the height field of the $a$-height function.
\end{definition}
In this section we predominately work with the $a$-height field since doing so simplifies various calculations. However since we are (in a sense) averaging over a large collection of points by pairing the height function and $a$-height function with test functions, we will see that results on correlations of the $a$-height field can be extended to the height field. For an example of this see the limit in \eqref{temp43rdz}.
We have a condition on test functions which we refer to in various places.
\begin{condition}\label{cond1}
For two test functions $f_i$, $i\in \{1,2\}$ there are functions $g_i,h_i\in C_c^\infty(\R^2)$ with
\begin{align}
f_i=\partial g_i+\overline{\partial}h_i
\end{align} 
and where both $g_1, h_1$ have disjoint support from $g_2, h_2$, that is 
\begin{align}
(\text{Supp}(g_1)\cup\text{Supp}(h_1))\cap(\text{Supp}(g_2)\cup\text{Supp}(h_2))=\emptyset.
\end{align}
\end{condition}
We use this condition in particular to call the formulas \eqref{temp04311}. The author believes the condition can be removed using ideas from the proof of Theorem 1.3 in \cite{BW} together with further analysis of the inverse Kasteleyn matrix.

We also define  for $j\in \{0,1\}$
\begin{align}
\partial_j h_\eps^a(f):=(2\eps)^2\sum_{x\in \eps(2\Z+1/2)^2} \partial_j h^a(x/\eps)f(x).
\end{align}
We also recall \eqref{diffops}, and define 
\begin{align}
\partial h_\eps^a(f)=\frac{1}{2}(-i\partial_0 h_\eps^a(f)+\partial_1 h_\eps^a(f)),&& \overline{\partial} h_\eps^a(f)=\frac{1}{2}(i\partial_0 h_\eps^a(f)+\partial_1 h_\eps^a(f)).
\end{align}
We have the following theorem on convergence of the two-point correlation functions of the weak directional derivatives of the height field to those of the sine-Gordon field.
\begin{theorem}\label{derivstosinegordon}
Let $\eps>0$ and $\lambda>0$ be fixed not depending on $\eps$. 
Also let $0<a<1$, $z\in\R\setminus\{0\}$ be such that 
\begin{align}
a=1-\lambda\eps, && |z|=\lambda\frac{e^{\gamma/2}}{4\sqrt{2}\pi},
\end{align}
where $\gamma$ is the Euler-Mascheroni constant.
Let $f_1,f_2:\R^2\rarrow \R$ be smooth functions with disjoint, compact support. We have the following limits, for $i,j\in\{0,1\}$
\begin{align}
\E_a[\partial_ih_\eps^a(f_1)\partial_jh_\eps^a(f_2)]\rarrow \frac{1}{4\pi}\E_{SG(4\pi,z)}[\partial_i\varphi(f_1)\partial_j\varphi(f_2)]\label{temp34edmn}
\end{align}
as $\eps\rarrow 0$.
\begin{proof}
We have the left hand side of \eqref{temp34edmn} as
\begin{align}
(2\eps)^4\sum_{x_1,x_2\in \eps(2\Z+1/2)^2} \E_a[\partial_ih^a(x_1/\eps) \partial_j h^a(x_2/\eps)]f(x_1)f(x_2).
\end{align}
Since $f_1,f_2$ have compact, disjoint support the limit follows from the following lemma combined with lemma \ref{dirderivSG}.
\begin{lemma}
Let $a=1-\lambda\eps>0$, $\eps>0$ and $\lambda>0$ is fixed not depending on $\eps$. Let $x_1,x_2$ be two $a$-faces such that
\begin{align}
x_1-x_2=(x,y)/\eps, \quad \text{ where } (x,y)\neq 0.
\end{align}
We have the following limits, 
\begin{align}
\E_a[\partial_1 h^a(x_1) \partial_1h^a(x_2)]&= \frac{\lambda^2}{4\pi^2}\Big(-K_0(\lambda z/\sqrt{2})^2+\frac{x^2-y^2}{z^2}K_1(\lambda z/\sqrt{2})^2\Big)+o(1),\label{temp3210}\\
\E_a[\partial_0 h^a(x_1) \partial_0h^a(x_2)]&= \frac{\lambda^2}{4\pi^2}\Big(-K_0(\lambda z/\sqrt{2})^2+\frac{y^2-x^2}{z^2}K_1(\lambda z/\sqrt{2})^2\Big)+o(1),\label{temp3212}\\
\E_a[\partial_1 h^a(x_1) \partial_0 h^a(x_2)]&=-\frac{\lambda^2xy}{2\pi^2z^2}K_1(\lambda z/\sqrt{2})^2+o(1),\label{temp3209}\\
\E_a[\partial_0 h^a(x_1) \partial_1 h^a(x_2)]&=-\frac{\lambda^2xy}{2\pi^2z^2}K_1(\lambda z/\sqrt{2})^2+o(1).\label{temp3211}
\end{align}
as $\eps$ tends to zero, where $z=\sqrt{x^2+y^2}$ and uniform for $(x,y)$ in a compact subset of $\R^2\setminus \{0\}$. 
\begin{proof}
In the following we identify edges and vertices under \eqref{changeofcoordsG} when necessary. We just prove \eqref{temp3210} and \eqref{temp3209} since \eqref{temp3212} can be found with the same type of argument and \eqref{temp3211} follows by symmetry from \eqref{temp3209}. We first  prove \eqref{temp3210}.

A straight path travelling from an $a$-face $x_1$ to an $a$-face $x_1+(0,2)$ first crosses an edge $e_1$ with its black vertex on the right and then crosses an edge $e_2$ with its white vertex on the right. Similarly, a straight path travelling from an $a$-face $x_2$ to $x_2+(0,2)$ first crosses an edge $e_3$ with its black vertex on the right and then crosses an edge $e_4$ with its white vertex on the right. Hence we have
\begin{align}
2\eps\partial_1h^a(x_1)=h^a(x_1+(2,0))-h^a(x_1)=(-\ind_{e_1\in \omega}+1/4)+(\ind_{e_2\in \omega}-1/4)=\ind_{e_2\in \omega}-\ind_{e_1\in\omega},
\end{align}
and
\begin{align}
2\eps\partial_1h^a(x_2)=\ind_{e_4\in\omega}-\ind_{e_3\in\omega}.
\end{align}
Substitution then gives
\begin{align}
4\eps^2\E_a[\partial_1h^a(x_1)\partial_1h^a(x_2)]=\P_a(e_1,e_3\in\omega)+\P_a(e_2,e_4\in\omega)-\P_a(e_2,e_3\in \omega)-\P_a(e_1,e_4\in\omega).
\end{align}
We know that dimers form a determinantal point process \eqref{dimersdetprocess}, and thus
\begin{align}
\P_a(e_i,e_j\in\omega)=\det(L(e_k,e_\ell))_{k,\ell\in\{i,j\}}=\P_a(e_i\in\omega)\P(e_j\in\omega)-L(e_i,e_j)L(e_j,e_i).\label{dimerdetprocess}
\end{align}
For an edge $e=(b,w)\in \tilde{B}_{\eps_2}\times \tilde{W}_{\eps_1}$, $\P_a(e\in\omega)=\K_a(b,w)\K_a^{-1}(w,b)$, and it is easy to see from the symmetries in \eqref{infinvKasteleyn} that $\P_a(e\in\omega)$ are all equal for the four types of edges $\eps_1,\eps_2$ of weight $a$. Hence
\begin{align}
4\eps^2\E_a[\partial_1h^a(x_1)\partial_1h^a(x_2)]=-L(e_1,e_3)L(e_3,e_1)-L(e_2,e_4)L(e_4,e_2)\\+L(e_2,e_3)L(e_3,e_2)+L(e_1,e_4)L(e_4,e_1).\nonumber
\end{align}
We have $\K_a(e_1)=\K_a(e_2)=\K_a(e_3)=\K_a(e_4)=ai=i+O(\eps)$. We label $e_j=(b_j,w_j)$, $j=1,...,4$, and we write
\begin{align}\label{temp095029}
\E_a[\partial_1h^a(x_1)\partial_1h^a(x_2)] &=(2\eps)^{-2}\K_a^{-1}(w_3,b_1)\K_a^{-1}(w_1,b_3)+(2\eps)^{-2}\K_a^{-1}(w_4,b_2)\K_a^{-1}(w_2,b_4)\\ & \ \ \ -(2\eps)^{-2}\K_a^{-1}(w_3,b_2)\K_a^{-1}(w_2,b_3)-(2\eps)^{-2}\K_a^{-1}(w_4,b_1)\K_a^{-1}(w_1,b_4)+O(\eps).\nonumber
\end{align}
We would now like to make use of theorem \ref{invkastasymp}. For concreteness we write  the image of the face 
\begin{align}
x_j=(-\beta_j\eps^{-1}+1/2,\alpha_j\eps^{-1}+1/2)\label{temp024h4}
\end{align} under \eqref{changeofcoordsG} as $(1+\alpha_j\eps^{-1})\vec{e}_1+\beta_j\eps^{-1}\vec{e}_2$ where $(\alpha_j\eps^{-1},\beta_j\eps^{-1})\in (2\Z)^2$.
 This gives the image of $x_j+(0,2)$ as  $(1+\alpha_j\eps^{-1})\vec{e}_1+(\beta_j-2\eps)\eps^{-1}\vec{e}_2$.
 By their definition, $e_1,e_3\in \tilde{B_1}\times \tilde{W_1}$ and $e_2,e_4\in \tilde{B_0}\times \tilde{W_0}$. 
 So for example $w_3$ and $b_1$ have coordinates
 \begin{align}
 (1+\alpha_2\eps^{-1})\vec{e}_1+\beta_2\vec{e}_2+(0,1) &&\text{and}&&\quad (1+\alpha_1)\eps^{-1}\vec{e}_1+\beta_1\vec{e}_2+(1,0)
 \end{align}
 under \eqref{changeofcoordsG}. We also have 
 \begin{align}
 \beta=\beta_2-\beta_1=-x&&\alpha=\alpha_2-\alpha_1=y.
 \end{align}
 Hence we can make use of theorem \ref{invkastasymp} to get \eqref{temp095029} as
 \begin{align}
 &\frac{1}{4}\Big(\frac{i\lambda}{2\pi}K_0(\lambda z/\sqrt{2})-\frac{i\lambda\beta}{\sqrt{2}\pi z}K_1(\lambda z/\sqrt{2})\Big)\Big(\frac{i\lambda}{2\pi}K_0(\lambda z/\sqrt{2})-\frac{i\lambda(-\beta)}{\sqrt{2}\pi z}K_1(\lambda z/\sqrt{2})\Big)\\
 &+\frac{1}{4}\Big(\frac{i\lambda}{2\pi}K_0(\lambda z/\sqrt{2})+\frac{i\lambda\beta}{\sqrt{2}\pi z}K_1(\lambda z/\sqrt{2})\Big)\Big(\frac{i\lambda}{2\pi}K_0(\lambda z/\sqrt{2})+\frac{i\lambda(-\beta)}{\sqrt{2}\pi z}K_1(\lambda z/\sqrt{2})\Big)\nonumber\\
&-\frac{1}{4}\Big(-\frac{\lambda}{2\pi}K_0(\lambda z/\sqrt{2})+\frac{\lambda\alpha}{\sqrt{2}\pi z}K_1(\lambda z/\sqrt{2})\Big)\Big(-\frac{\lambda}{2\pi}K_0(\lambda z/\sqrt{2})-\frac{\lambda(-\alpha)}{\sqrt{2}\pi z}K_1(\lambda z/\sqrt{2})\Big)\nonumber\\
&-\frac{1}{4}\Big(-\frac{\lambda}{2\pi}K_0(\lambda z/\sqrt{2})-\frac{\lambda\alpha}{\sqrt{2}\pi z}K_1(\lambda z/\sqrt{2})\Big)\Big(-\frac{\lambda}{2\pi}K_0(\lambda z/\sqrt{2})+\frac{\lambda(-\alpha)}{\sqrt{2}\pi z}K_1(\lambda z/\sqrt{2})\Big)\nonumber\\
&+o(1).\nonumber\\
=& \ \frac{\lambda^2}{4\pi^2}\Big(-K_0(\lambda z/\sqrt{2})^2+\frac{x^2-y^2}{z^2}K_1(\lambda z/\sqrt{2})^2\Big)+o(1)\nonumber
 \end{align}
 where $z=\sqrt{\alpha^2+\beta^2}=\sqrt{x^2+y^2}$.
 
 Now we prove \eqref{temp3209}. We will redefine the edges $e_3,e_4$ for the remainder of the proof.  That is $e_1,e_2$ are defined above in the sense that a straight path travelling from an $a$-face $x_1$ to an $a$-face $x_1+(0,2)$ first crosses an edge $e_1$ with its black vertex on the right and then crosses an edge $e_2$ with its white vertex on the right. However now, a straight path travelling from an $a$-face $x_2$ to $x_2+(2,0)$ (i.e moving \emph{horizontally}) first crosses an edge $e_3$ with its white vertex on the right and then crosses an edge $e_4$ with its black vertex on the right. So we have 
 \begin{align}
 2\eps\partial_1h^a(x_1)=\ind_{e_2}-\ind_{e_1} && 2\eps\partial_0h^a(x_2)=\ind_{e_3}-\ind_{e_4}.
 \end{align}
 Following the previous argument, one obtains 
 \begin{align}
 4\eps^2\E_a[\partial_1h^a(x_1)\partial_0h^a(x_2)]=&   -L(e_1,e_4)L(e_4,e_1)-L(e_2,e_3)L(e_3,e_2)\\
& \  +L(e_2,e_4)L(e_4,e_2)+L(e_1,e_3)L(e_3,e_1).\nonumber
 \end{align}
 Now $\K_a(e_1),\K_a(e_2)=ai$, $\K_a(e_3)=\K_a(e_4)=a$. We have the labels $e_j=(b_j,w_j)$, $j=1,...,4$, and we get
\begin{align}\label{temp09239ss}
\E_a[\partial_1h^a(x_1)\partial_0 h^a(x_2)] &=-(2\eps)^{-2} \ i \ \K_a^{-1}(w_4,b_1)\K_a^{-1}(w_1,b_4)-(2\eps)^{-2} \ i \   \K_a^{-1}(w_3,b_2)\K_a^{-1}(w_2,b_3)\\ & \ \ \ +(2\eps)^{-2} \ i \ \K_a^{-1}(w_4,b_2)\K_a^{-1}(w_2,b_4)+(2\eps)^{-2} \ i \ \K_a^{-1}(w_3,b_1)\K_a^{-1}(w_1,b_3)+O(\eps).\nonumber
\end{align}
By their definitions,  $e_1\in \tilde{B}_1\times \tilde{W}_1$, $e_2\in \tilde{B}_0\times \tilde{W}_0$, $e_3\in \tilde{B}_1\times \tilde{W}_0$ and $e_4\in \tilde{B}_0\times \tilde{W}_1$. We assign the coordinates \eqref{temp024h4} and use theorem \ref{invkastasymp} again.
This gives \eqref{temp09239ss} as
\begin{align}
&-i\Big(\frac{i\lambda}{2\pi}K_0(\lambda z/\sqrt{2})-\frac{i\lambda\beta}{\sqrt{2}\pi z}K_1(\lambda z/\sqrt{2})\Big)\Big(-\frac{\lambda}{2\pi}K_0(\lambda z/\sqrt{2})+\frac{\lambda(-\alpha)}{\sqrt{2}\pi z}K_1(\lambda z/\sqrt{2})\Big)\\
&-i\Big(\frac{i\lambda}{2\pi}K_0(\lambda z/\sqrt{2})+\frac{i\lambda\beta}{\sqrt{2}\pi z}K_1(\lambda z/\sqrt{2})\Big)\Big(-\frac{\lambda}{2\pi}K_0(\lambda z/\sqrt{2})-\frac{\lambda(-\alpha)}{\sqrt{2}\pi z}K_1(\lambda z/\sqrt{2})\Big)\nonumber\\
&+i\Big(-\frac{\lambda}{2\pi}K_0(\lambda z/\sqrt{2})+\frac{\lambda\alpha}{\sqrt{2}\pi z}K_1(\lambda z/\sqrt{2})\Big)\Big(\frac{i\lambda}{2\pi}K_0(\lambda z/\sqrt{2})+\frac{i\lambda(-\beta)}{\sqrt{2}\pi z}K_1(\lambda z/\sqrt{2})\Big)\nonumber\\
&+i\Big(-\frac{\lambda}{2\pi}K_0(\lambda z/\sqrt{2})-\frac{\lambda\alpha}{\sqrt{2}\pi z}K_1(\lambda z/\sqrt{2})\Big)\Big(\frac{i\lambda}{2\pi}K_0(\lambda z/\sqrt{2})-\frac{i\lambda(-\beta)}{\sqrt{2}\pi z}K_1(\lambda z/\sqrt{2})\Big)\nonumber\\
& +o(1)\nonumber\\
& =-\frac{\lambda^2 xy}{2\pi^2z^2}K_1(\lambda z/\sqrt{2})^2+o(1)\nonumber
\end{align}
which proves \eqref{temp3209}.
\end{proof}
\end{lemma}
\end{proof}
\end{theorem}
By the Coleman correspondence \cite{BW}, we have the following convergence of two-point correlations of weak derivatives of the height field to truncated correlations of the free massive Dirac fermions.
\begin{corollary}
Recall the correlations functions of the free massive Dirac fermions with mass $\mu>0$, \eqref{fermcorrpaired}.
Let $\eps>0$ and $\mu>0$ not depend on $\eps$. 
Also let $0<a<1$ and the mass $\mu$ be related by
\begin{align}
a=1-\mu\sqrt{2}\eps.
\end{align}
For smooth functions $f_1^+,f_2^+,f_1^-,f_2^-$ all with disjoint, compact support we have
\begin{align}
\E_a[\prod_{j=1}^q(-i\partial h_\eps^a(f_j^+))\prod_{j'=1}^{q'}(-i\overline{\partial}h_\eps^a(f_{j'}^-))]\rarrow \frac{1}{4}\Big \langle\prod_{j=1}^q\overline{\psi}_2\psi_1(f_j^+)\prod_{j'=1}^{q'}\overline{\psi}_1\psi_2(f_{j'}^-)\Big \rangle_{FF(\mu)}^T
\end{align}
as $\eps\rarrow 0$ with $q,q'\in\{0,1\}$ such that $q+q'=2$ and where $h_\eps^a$ is defined above.
\end{corollary}

We introduce the following notation, we can write the $a$-height function as
\begin{align}
h^a(y)=\sum_{e,e'\in \gamma(y)}\sigma_e (\ind_e-\ind_{e'}), \label{aheightrep}
\end{align}
where $\gamma(y)$ is a path from $(1/2,1/2)$ to $y\in (2\Z+1/2)^2$ made of straight vertical and horizontal line segments of length two, 
the sum is over edges crossed by $\gamma(y)$, and $e$ ($e'$) is the first (second) edge crossed by $\gamma(y)$ when travelling from an $a$-face to one of the four nearest $a$-faces.
Also $\sigma_e=1$ if $\gamma(y)$ crosses $e$ with a white vertex on the right and $\sigma_e=-1$ otherwise.

We have the following Proposition on the convergence of the two-point correlation function of the $a$-height field to the two-point correlation function of the sine-Gordon field.
\begin{proposition}\label{aheighttwoptcorr}
Let $\eps>0$ and $\lambda>0$ be fixed not depending on $\eps$. 
Also let $0<a<1$, $z\in\R\setminus\{0\}$ be such that 
\begin{align}
a=1-\lambda\eps, && |z|=\lambda\frac{e^{\gamma/2}}{4\sqrt{2}\pi},
\end{align}
where $\gamma$ is the Euler-Mascheroni constant.
Let $f_i\in C_c^\infty(\R^2)$, $i\in\{1,2\}$ be two functions which satisfy condition \ref{cond1}.
We have the following limit,
\begin{align}
\E_a[h_\eps^a(f_1)h_\eps^a(f_2)]\rarrow \frac{1}{4\pi}\E_{SG(4\pi,z)}[\varphi(f_1)\varphi(f_2)]\label{temp52df4}
\end{align}
as $\eps\rarrow 0$. 
\begin{proof}
We have 
\begin{align}\label{temp31440}
\E_a[h_\eps^a(f_1)h_\eps^a(f_2)]=& \ \E_a[h_\eps^a(\partial g_1)h_\eps^a(\partial g_2)]+\E_a[h_\eps^a(\overline{\partial} h_1)h_\eps^a(\overline{\partial} h_2)]\\
&+\E_a[h_\eps^a(\partial g_1)h_\eps^a(\overline{\partial} h_2)]+\E_a[h_\eps^a(\overline{\partial} h_1)h_\eps^a(\partial g_2)].\nonumber
\end{align}
Let us focus on $\E_a[h_\eps^a(\partial g_1)h_\eps^a(\partial g_2)]$ and $\E_a[h_\eps^a(\partial g_1)h_\eps^a(\overline{\partial} h_2)]$. 
Define
\begin{align}
\delta_0^\eps f(x)=\frac{f(x)-f(x-\eps(2,0))}{2\eps}, && \delta_1^\eps f(x)=\frac{f(x)-f(x-\eps(0,2))}{2\eps}.
\end{align}
Observe that via summation by parts, 
\begin{align}
\partial_j h_\eps^a(f)&=-(2\eps)^2\sum_{x\in \eps(2\Z+1/2)^2} h^a(x/\eps)\delta_j^\eps f(x)\\
&=-(2\eps)^2\sum_{x\in \eps(2\Z+1/2)^2} h^a(x/\eps)(\partial_j f(x)+R(x,\eps)\eps).\nonumber
\end{align}
where $R(x,\eps)$ is a smooth boundedly supported function for each $\eps$ that comes from the remainder in Taylor's theorem, and which is a bounded function over $(x,\eps)$, and $\partial_jf(x)$ means the directional derivative of $f\in C_c^\infty(\R^2)$ in the direction $x_j$, $x=(x_0,x_1)$. Hence
\begin{align}
\label{temp0505}\E_a[h_\eps^a(\partial g_1)h_\eps^a(\partial g_2)]= & \ \E_a[\partial h_\eps^a( g_1)\partial h_\eps^a(g_2)]-\eps\E_a[h_\eps^a( R_1(\cdot,\eps))\partial h_\eps^a(g_2)+\partial h_\eps^a( g_1)h_\eps^a(R_2(\cdot,\eps))]\\
&+\eps^2\E_a[h_\eps^a(R_1(\cdot,\eps))h_\eps^a(R_2(\cdot,\eps))].\nonumber\\
& = :\ \E_a[\partial h_\eps^a( g_1)\partial h_\eps^a(g_2)]+ \eps \tilde{R}_1(\eps).\nonumber
\end{align}
Similarly,
\begin{align}
E_a[h_\eps^a(\partial g_1)h_\eps^a(\overline{\partial} h_2)]=\E_a[\partial h_\eps^a( g_1)\overline\partial h_\eps^a(h_2)]+\eps\tilde{R}_2(\eps).\label{temp302f}
\end{align}
From Theorem \ref{derivstosinegordon}
\begin{align}\label{temp09523}
\E_a[\partial h_\eps^a( g_1)\partial h_\eps^a(g_2)]\rarrow \frac{1}{4\pi}\E_{SG(4\pi,z)}[\partial \varphi(g_1)\partial \varphi(g_2)],\\
\E_a[\partial h_\eps^a( g_1)\overline\partial h_\eps^a(h_2)]\rarrow \frac{1}{4\pi}\E_{SG(4\pi,z)}[\partial \varphi(g_1)\overline\partial \varphi(h_2)]\nonumber
\end{align}
as $\eps\rarrow 0$.
The following formulas (see equations (3.36), (3.43) in \cite{BW}) can be derived via residue calculus
\begin{align}\label{temp04311}
\E_{SG(4\pi,z)}[\partial \varphi(g_1)\partial \varphi(g_2)]= \int_{\R^2}\frac{dp}{(2\pi)^2}\widehat{\partial g_1}(p)\widehat{\partial g_2}(-p)C_{A|z|}(p),\\
\E_{SG(4\pi,z)}[\partial \varphi(g_1)\overline\partial \varphi(h_2)]= \int_{\R^2}\frac{dp}{(2\pi)^2}\widehat{\partial g_1}(p)\widehat{\overline\partial h_2}(-p)C_{A|z|}(p)\nonumber
\end{align}
with the integrals understood in the principal value sense. By taking complex conjugates, we can use \eqref{temp0505}, \eqref{temp302f}, \eqref{temp09523} and \eqref{temp04311} in \eqref{temp31440} to get \eqref{temp52df4}. 
\\
All that remains is to show that $\eps \tilde{R}_i(\eps)\rarrow 0$, $i=1,2$. This is quite straightforward. Consider the term $\eps\E_a[h_\eps^a(R_1(\cdot,\eps)\partial h_\eps^a(g_2)]$.
Using the notation in \eqref{aheightrep}, let $\partial h_\eps^a(x_2/\eps)=\sigma_{e_1}(\ind_{e_1}-\ind_{e_1'})/\eps$ so
\begin{align}
&\E_a[h_\eps^a(R_1(\cdot,\eps)\partial h_\eps^a(g_2)]\label{temp4r50}\\
&= \E_a[(2\eps)^3\sum_{x_1,x_2\in \eps(2\Z+1/2)^2}\sum_{e_2,e_2'\in \gamma(x_1/\eps)}\sigma_{e_1}\sigma_{e_2}(\ind_{e_1}-\ind_{e_1'})(\ind_{e_2}-\ind_{e_2'})R_1(x_1,\eps)g_2(x_2)].
\nonumber
\end{align}
Also one can take the path $\gamma(x_1)/\eps$ to have length $<C'/\eps$, so the modulus of \eqref{temp4r50} is bounded above by 
\begin{align}
(2\eps)^2\sum_{x_1,x_2\in\eps(2\Z+1/2)^2}\sup_{e_2,e_2'\in \gamma(x_1/\eps)}|\E_a[(\ind_{e_1}-\ind_{e_1'})(\ind_{e_2}-\ind_{e_2'})]R_1(x_1,\eps)g_2(x_2)|.\label{temp0923k}
\end{align}

For all $\eps$ sufficiently small, both $R_i(\cdot,\eps)$ and $g_i(\cdot)$ have support contained in a compact set $D_i^\eps$ where $D_1^\eps\cap D_2^\eps=\emptyset$. So one sees that the edges $e_1,e_1'$ are distance $>C\eps$ from the edges $e_2,e_2'$ in the above sum which justifies the use of Theorem \ref{invkastasymp}.  Now with arguments very similar to those in the proof of Theorem \ref{derivstosinegordon} and by the boundedness of the Bessel functions on compact sets, one can see that \eqref{temp0923k} is bounded in $\eps$,  so $\eps\E_a[h_\eps^a(R_1(\cdot,\eps)\partial h_\eps^a(g_2)]\rarrow 0$, the other terms in $\eps \tilde{R}(\eps)$ are treated similarly.
\end{proof}
\end{proposition}
We can now prove the main result of this paper on the convergence of the two-point correlation function of the height field to the two-point correlation function of the sine-Gordon field (Theorem \ref{mainthm} restated).
\begin{theorem}\label{mainthm2}
Let $\eps>0$ and $z\in\R\setminus\{0\}$ be fixed not depending on $\eps$. Let $0<a<1$, be such that 
\begin{align}
a=1-C |z|  \eps
\end{align}
where $C =4\sqrt{2} \pi e^{-\gamma/2}$ and $\gamma$ is the Euler-Mascheroni constant.
We have the following limit,
\begin{align}
\E_a[h_\eps(f_1)h_\eps(f_2)]\rarrow \frac{1}{4\pi}\E_{SG(4\pi,z)}[\varphi(f_1)\varphi(f_2)]\label{temp52df4q}
\end{align}
as $\eps\rarrow 0$. \\
In the above, $f_i\in C_c^\infty(\R^2)$, $i\in\{1,2\}$ are two functions satisfying condition \ref{cond1}.
\begin{proof}
We use Proposition \ref{aheighttwoptcorr}, together with an argument that shows the $a$-height field $h_\eps^a$ is "close enough" to the height field $h_\eps$, that is, we show
\begin{align}
\E_a[h_\eps(f_1)h_\eps(f_2)]-\E_a[h_\eps^a(f_1)h_\eps^a(f_2)]\rarrow 0 \label{temp43rdz}
\end{align} 
as $\eps\rarrow 0$.
Write the $16$ points $\{(0,0),\eps(1,0),\eps(0,1),\eps(1,1)\}^2$ as $A$. We can write the left-hand side of \eqref{temp43rdz} as
\begin{align}
\eps^4\sum_{x_1,x_2\in \eps(2\Z+1/2)^2}\Big(\sum_{y_1,y_2\in A} \E_a[h((x_1+y_1)/\eps)h((x_2+y_2)/\eps)]f(x_1+y_1)f(x_2+y_2)\Big)\label{temp23fr}\\-4^2\E_a[h(x_1/\eps)h(x_2/\eps)]f(x_1)f(x_2).\nonumber
\end{align}
By Taylor's theorem we can get a bounded function $R(\eps,x_1,x_2,y_1,y_2)$ with bounded support such that
\begin{align}
f(x_1+y_1)f(x_2+y_2)=f(x_1)f(x_2)+\eps R(\eps,x_1,x_2,y_1,y_2),
\end{align}
which we use in \eqref{temp23fr}. Hence we will be done if we show that both of the limits
\begin{align}
\eps \E_a[h((x_1+y_1)/\eps)h((x_2+y_2)/\eps)]R(\eps,x_1,x_2,y_1,y_2)\rarrow 0\label{temp23dcz}\\
\E_a[h((x_1+y_1)/\eps)h((x_2+y_2)/\eps)-h(x_1/\eps)h(x_2/\eps)]\rarrow 0 \label{temp23dc2}
\end{align} 
hold uniformly for $x_1, x_2\in \eps(2\Z+1/2)^2$ in the supports of $f_1,f_2$ respectively and for $y_1,y_2\in A$ as $\eps\rarrow 0$. To see \eqref{temp23dcz} we argue that  $\E_a[h((x_1+y_1)/\eps)h((x_2+y_2)/\eps)]$ is bounded via the use of Theorem \ref{invkastasymp} and boundedness of the Bessel functions on compact sets.  Since this is similar to before we just focus on showing \eqref{temp23dc2}.
As in the notation introduced in \eqref{aheightrep}, write
\begin{align}
h((x_i+y_i)/\eps)=h(x_i/\eps)+\sum_{e\in \gamma_{i}}\sigma_e(\ind_e-1/4) \label{temp23vxb}
\end{align}
where the path $\gamma_{i}$ is the shortest path made of vertical and horizontal line segments from the point $x_i/\eps$ to the point $(x_i+y_i)/\eps$ and the sum is over edges crossed by the path, observe $\text{length}(\gamma_i)\leq 2$. Insert \eqref{temp23vxb} into the left hand side of \eqref{temp23dc2}, we get
\begin{align}
\E_a[\sum_{e\in \gamma_{1}}\sigma_e(\ind_e-1/4) h(x_2)+\sum_{e\in \gamma_{2}}\sigma_e(\ind_e-1/4) h(x_1)]\\+ \ 
\E_a[\sum_{e_1\in \gamma_{1}}\sum_{e_2\in \gamma_{2}}\sig_{e_1}\sig_{e_2}(\ind_{e_1}-1/4)(\ind_{e_2}-1/4)].\label{temp365g}
\end{align}
Note that $h(x_i)=h^a(x_i)$. Using the symmetries in \eqref{infinvKasteleyn} we recall that that $\P_a(e\in\omega)$ are all equal for the four types of edges $\eps_1,\eps_2$ of weight $a$ and so $\E_a[h(x_i)]=\E_a[h^a(x_i)]=0$. We then have 
\begin{align}
|\E_a[\sum_{e\in \gamma_{1}}\sigma_e(\ind_e-1/4) h(x_2)]|  =& \ |\E_a[\sum_{e\in \gamma_{1}}\sigma_e\ind_e h(x_2)]|\nonumber \\ \leq & \ \sum_{e\in \gamma_1}| \E_a[\ind_eh(x_2)]|\nonumber\\
 =& \ \sum_{e\in \gamma_1}| \E_a[\ind_e \sum_{e_2,e_2'\in\gamma(x_2/\eps)}\sigma_{e_2} (\ind_{e_2}-\ind_{e_2'})]|.\label{temp23bnv}
\end{align}
We can take the path $\gamma(x_2/\eps)$ so that the number of edges it crosses is bounded above by $\tilde{C}/\eps$ for some $\tilde{C}>0$, to get \eqref{temp23bnv} less than or equal to
\begin{align}
 C\frac{\tilde{C}}{\eps}\sup_{e\in \gamma_1 , \ e_2,e_2'\in \gamma(x_2/\eps)\cap \text{supp}(f_2)}|\P(e,e_2)-\P(e,e_2')|.\label{temp56v}
\end{align}
Now from \eqref{dimerdetprocess} and Proposition \ref{invkastasymp} we see \eqref{temp56v} is $O(\eps)$. Now we prove the term in \eqref{temp365g} goes to zero. We have
\begin{align}
\E_a[(\ind_{e_1}-1/4)(\ind_{e_2}-1/4)]=\P(e_1,e_2)-\frac{1}{4}(\P(e_1)+\P(e_2))+\frac{1}{16}. \label{temp35rf2}
\end{align}
From \cite{Bain} we have $\P(e\in \omega)=1/4+O(\eps\ln\eps)$ (see the equation preceding equation 3.7).
From \eqref{dimerdetprocess} and Theorem \ref{invkastasymp} we have $\P(e_1,e_2)=\P(e_1)\P(e_2)-L(e_1,e_2)L(e_2,e_1)=1/16+O(\eps \ln\eps)+O(\eps^2)$. Hence \eqref{temp35rf2} is $O(\eps\ln\eps)$ and we are done.
\end{proof}
\end{theorem}
\section{Asymptotics of $\K_{a}^{-1}$}\label{sectionAsymptotics}

\subsection{Definitions and an expression for $\K_a^{-1}$}
In this subsection we obtain a formula for $\K_a^{-1}$ (given in Lemma \ref{preasymplem}) which is the starting point for the asymptotics of the next section.
We will use the single integral formula for $\K_{a}^{-1}$ derived in paper \cite{C/J} which we give below in \eqref{K_aalt}. It is instructive to derive this single integral formula from the double integral \eqref{infinvKasteleyn}. We include a condensed version of their method starting from the formula \eqref{infinvKasteleyn} for pedagogical purposes, for more details see the original paper. 

Consider the inverse Kasteleyn matrix $\K_a^{-1}(x,y)$, where $(x,y)\in W_{\eps_1}\times B_{\eps_2}$  and $u(2\vec{e}_1)+v(2\vec{e}_2)$, $u,v\in \Z$, is the translation to get from the fundamental domain containing $x$ to the fundamental domain containing $y$. We recall formula \eqref{infinvKasteleyn} as
\begin{align}
\K_{a}^{-1}(x,y)=\frac{1}{(2\pi i)^2}\int_{\Gamma_{1}}\frac{dz}{z}\int_{\Gamma_{1}}\frac{dw}{w}\frac{Q(z,w)_{\eps_1+1,\eps_2+1}}{P(z,w)} z^uw^v\label{temp3rfs}
\end{align}
and $P(z,w)=-2-2a^2-a/w-aw-a/z-az$ for the readers convenience. Observe that since $0<a<1$, $P(z,w)$ has no zeros on $\Gamma_1^2$. Making the change of variables 
$z=u_2/u_1$ for $u_1$ fixed and then $w=u_1u_2$, $u_1,u_2\in \Gamma_1$ in \eqref{temp3rfs}, we write
\begin{align}
\tilde{Q}(u_1,u_2):=Q(u_2/u_1,u_1u_2)=\begin{pmatrix}
i(a+u_1u_2) & -(a+u_2/u_1)\\
-(a+u_1/u_2) & i(  a+1/(u_1u_2)   )
\end{pmatrix}
\end{align}
and
\begin{align}
\tilde{P}(u_1,u_2):=P(u_2/u_1,u_1u_2)=-2(1+a^2)(1+\frac{c}{2}(u_1+1/u_1)(u_2+1/u_2)),\label{temp35rf}
\end{align}
where
\begin{align}
c=\frac{a}{1+a^2}\in(0,1/2).\nonumber
\end{align}

Let $k=u+v$ and $\ell=v-u$ so that we have
\begin{align}
\K_{a}^{-1}(x,y)&=\frac{1}{(2\pi i)^2}\int_{\Gamma_{1}}\frac{du_1}{u_1}\int_{\Gamma_{1}}\frac{du_2}{u_2}\frac{\tilde{Q}(u_1,u_2)_{\eps_1+1,\eps_2+1}}{\tilde{P}(u_1,u_2)}u_1^{v-u}u_2^{v+u}\label{etmp45trfd|}\\
&=  \begin{pmatrix} ia \ \mathcal{I}_a(k,\ell)+i\ \mathcal{I}_a(k+1,\ell+1) & -a \ \mathcal{I}_a(k,\ell)-\mathcal{I}_a(k+1,\ell-1)\\
-a \ \mathcal{I}_a(k,\ell)-\mathcal{I}_a(k-1,\ell+1)& ia \ \mathcal{I}_a(k,\ell)+i \ \mathcal{I}_a(k-1,\ell-1)\end{pmatrix}_{\eps_1+1,\eps_2+1}\nonumber
\end{align}
where we define
\begin{align}
\mathcal{I}_a(k,\ell):=\frac{1}{(2\pi i)^2}\int_{\Gamma_{1}}\frac{du_1}{u_1}\int_{\Gamma_{1}}\frac{du_2}{u_2}\frac{u_1^\ell u_2^k}{\tilde{P}(u_1,u_2)}=\mathcal{I}_a(|k|,|\ell|)\label{ikl1}
\end{align}
for integers $k,\ell$. The last equality follows from using symmetries of the form $u_i\rarrow 1/u_i$ in the integrand of $\mathcal{I}_a(k,\ell)$. 

We focus on getting a single integral formula for $\mathcal{I}_a(k,\ell)$.
Deform both of the contours over $\Gamma_1$ in \eqref{ikl1} to $\Gamma_R$  for $R<1$ very close to 1 (avoiding the zeros of $\tilde{P}$), so that
\begin{align}
\mathcal{I}_a(k,\ell)=\frac{1}{(2\pi i)^2}\int_{\Gamma_{R}}\frac{du_1}{u_1}\int_{\Gamma_{R}}\frac{du_2}{u_2}\frac{u_1^{|\ell|} u_2^{|k|}}{\tilde{P}(u_1,u_2)}.\label{ikl2}
\end{align}
Denote the punctured open unit disc $\D^*=\D\setminus\{0\}\subset \C$. Introduce the analytic bijective map 
\begin{align}
J:\D^*\rarrow \C\setminus i[-\sqrt{2c},\sqrt{2c}]; \ u\mapsto \sqrt{\frac{c}{2}}(u-1/u),
\end{align}
with analytic inverse
\begin{align}\label{Gfunction}
G:\C\setminus i[-\sqrt{2c},\sqrt{2c}] \rarrow \D^*;  \ w\mapsto \frac{1}{\sqrt{2c}}(w-\sqrt{w^2+2c})
\end{align}
where $\sqrt{w^2+2c}=i\sqrt{-\sqrt{2c}-iw}\sqrt{\sqrt{2c}-iw}$ and the two previous square roots are principal branch square roots. Equivalently,
\begin{align}
G(w)=\frac{w}{\sqrt{2c}}(1-\sqrt{1+\frac{2c}{w^2}})\label{temp34reflkj}
\end{align}
where the previous square root is the principle branch square root. We note $J$ is related to the Joukovski map and the above claims about $J$ and $G$ follow from chapter 6 in \cite{Mar}. Observe that
\begin{align}
\tilde{P}(u_1,u_2)=-2(1+a^2)(1-J(iu_1)J(iu_2)).\label{temp5gg4tf}
\end{align}

Making the change of variables $u_1=i^{-1}G(w_1), u_2=i^{-1}G(w_2)$ in \eqref{ikl2}, a short calculation gives
\begin{align}
\frac{du_i}{u_i}=\frac{dw_i}{\sqrt{w^2+2c}}, \ i=1,2,
\end{align}
and 
\begin{align}
\tilde{P}(i^{-1}G(w_1),i^{-1}G(w_2))=-2(1+a^2)(1-w_1w_2)
\end{align}
which follows from \eqref{temp5gg4tf} and the fact that $G$ is the inverse of $J$. 
After this change of variables we have
\begin{align}
\mathcal{I}_a(k,\ell)=\frac{-i^{-|k|-|\ell|}}{2(1+a^2)(2\pi i)^2}\int_{\gamma_{R}}dw_1\int_{\gamma_{R}}dw_2\frac{G(w_1)^{|\ell|}G(w_2)^{|k|}}{\sqrt{w_1^2+2c}\sqrt{w_2^2+2c}(1-w_1w_2)}.
\end{align}
The contour $\gamma_R$ is the image of $\Gamma_R$ under $J(i \ \cdot)$, so it is an ellipse. Infact, with some basic analysis one can show that for $R$ close enough to $1$, for each $w_1\in \gamma_R$, the pole at $1/w_1$ lies outside the region enclosed by the $w_2$ contour $\gamma_R$. Hence we deform the $w_2$ contour to $\Gamma_{R'}$ and take the limit $R'\rarrow \infty$. We pick up a single integral coming from the simple pole at $w_2=1/w_1$. The double integral vanishes in the limit since from \eqref{temp34reflkj} it is clear that for $|w_2|=R'$, $R'>>1$ very large, we have $|G(w_2)|=O(1/R')$. Finally deform the $w_1$ contour to $\Gamma_1$ in the single integral to get the formula
\begin{align}
\mathcal{I}_a(k,\ell)=\frac{-i^{-|k|-|\ell|}}{2(1+a^2)(2\pi i)^2}\int_{\Gamma_{1}}\frac{dw_1}{w_1}\frac{G(w_1)^{|\ell|}G(1/w_1)^{|k|}}{\sqrt{w_1^2+2c}\sqrt{w_1^2+2c}}.\label{ikl}
\end{align}
Plugging this back into \eqref{etmp45trfd|} we see we have derived a single integral formula for $\K_a^{-1}$. For example, if we take the two vertices $x,y\in W_0\times B_0$ in the same fundamental domain, i.e. $u=v=0$ so that $e=(y,x)$ is an $a$ weighted dimer, we have
\begin{align}
\K_a^{-1}(x,y)=\frac{-1}{2(1+a^2)2\pi i}\int_{\Gamma_{1}}\frac{dw_1}{w_1}\frac{i(a-G(w_1)G(1/w_1))}{\sqrt{w_1^2+2c}\sqrt{1/w_2^2+2c}}.
\end{align}

We now define some formulas which allow us to write a compact expression for $\K_a^{-1}(x,y)$ directly in terms of the planar coordinates of the vertices $x$ and $y$.
For $\eps_1,\eps_2\in\{0,1\}$, we write
\begin{equation}
h(\eps_1,\eps_2)=\eps_1(1-\eps_2)+\eps_2(1-\eps_1).\label{HHH}
\end{equation}
 Let $(x_1,x_2)\in W_{\eps_1}$, $(y_1,y_2)\in B_{\eps_2}$ and define
\begin{align}
&k_1=\frac{x_2-y_2-1}{2}+h(\eps_1,\eps_2), && \ell_1=\frac{y_1-x_1-1}{2},\label{klh1}\\
& k_2=k_1+1-2h(\eps_1,\eps_2) , && \ell_2=\ell_1+1.
\label{klh2}
\end{align}
For vertices $(x_1,x_2)\in W_{\eps_1}$, $(y_1,y_2)\in B_{\eps_2}$, we have the following single integral representation of the inverse Kasteleyn matrix (equation (4.22) in \cite{C/J}):
\begin{align}
\K_{a}^{-1}(x_1,x_2,y_1,y_2)=i^{1+h(\eps_1,\eps_2)}(a^{\eps_2}\mathcal{I}_a(k_1,\ell_1)+a^{1-\eps_2}\mathcal{I}_a(k_2,\ell_2))
\label{K_aalt}
\end{align}
with $\mathcal{I}_a$ as in \eqref{ikl}. Note that the function $E_{k,\ell}$ in equation (4.22) in \cite{C/J} is $-\mathcal{I}_a(k,\ell)$.
The formula \eqref{K_aalt} is just a compact way of writing \eqref{etmp45trfd|}. To see this, let $FD(v)$ denote the centre of the fundamental domain containing the vertex $v$, if $x=(x_1,x_2)\in W_{\eps_1}, y=(y_1,y_2)\in B_{\eps_2}$ then
\begin{align}
FD(x)=(x_1,x_2)-(2\eps_1-1,0), && FD(y)=(y_1,y_2)-(0,2\eps_2-1).
\end{align}
The translation to get from the fundamental domain containing $x$ to the fundamental domain containing $y$ is
\begin{align}
u(2\vec{e}_1)+v(2\vec{e}_2)=FD(y)-FD(x),
\end{align}
which implies
\begin{align}
2(-\ell,k)=2(u-v,u+v)=(y_1-x_1+2\eps_1-1,y_2-x_2-(2\eps_2-1)).
\end{align}
Substituting $\ell$ and $k$ into \eqref{etmp45trfd|} and using $\mathcal{I}(k,\ell)=\mathcal{I}(|k|,|\ell|)$ one can check that \eqref{K_aalt} and \eqref{etmp45trfd|} are equal in the four cases $\eps_1,\eps_2\in\{0,1\}$.
For later use, we note the symmetries
\begin{align}
\overline{\sqrt{w^2+2c}}=\sqrt{\overline{w}^2+2c}, && -\sqrt{w^2+2c}=\sqrt{(-w)^2+2c}
\label{squarerootsymmetries}
\end{align}
which give
\begin{align}
&\overline{G(w)}=G(\overline{w}),&& -G(w)=G(-w).
\label{Gsymmetries}
\end{align}

Consider two $a$-dimers, $e_j=(x(j),y(j))\in W_{\eps_1}\times B_{\eps_2}$ where their coordinates are specified by
\begin{align}
x(j)=(x_1(j),x_2(j))=(1+\alpha_j\eps^{-1})\vec{e}_1+\beta_j\eps^{-1}\vec{e}_2+(0,2\eps_1-1)\label{temp0392}\\
y(j)=(y_1(j),y_2(j))=(1+\alpha_j\eps^{-1})\vec{e}_1+\beta_j\eps^{-1}\vec{e}_2+(2\eps_2-1,0)\nonumber
\end{align}
where $j\in\{0,1\}$, $(\alpha_j\eps^{-1},\beta_j\eps^{-1})\in (2\Z)^2$, $\vec{e}_1=(1,1)$ and $\vec{e}_2=(-1,1)$. Also for the remainder of the article, $(\alpha_0,\beta_0)\neq (\alpha_1,\beta_1)$ and $\alpha_j,\beta_j$ are independent of $\eps$. Next we introduce some notation that will come into an integral representation of $\K_a^{-1}$ amenable to asymptotic analysis.
For $i\in \{0,1\}$, $i\neq j$, define 
\begin{align}
\alpha=\alpha_j-\alpha_i, && \beta=\beta_j-\beta_i,
\end{align}
\begin{align}\label{sig1sig2}
\sig_1=
\begin{cases}
\text{sign}(\alpha+\beta), & \text{if }\alpha\neq -\beta\\
-1,& \text{if }\alpha=-\beta \text{ and } \eps_1=0\\
1,& \text{if }\alpha=-\beta \text{ and } \eps_1=1
\end{cases},
&&
\sig_2=
\begin{cases}
\text{sign}(\alpha-\beta), & \text{if }\alpha\neq \beta\\
1,& \text{if }\alpha=\beta \text{ and } \eps_2=0\\
-1,& \text{if }\alpha=\beta \text{ and } \eps_2=1
\end{cases}.
\end{align}
Also define
\begin{align}\label{gfn}
g_{\eps_1,\eps_2}^{\sig_1,\sig_2}(w)=& \ a^{\eps_2}\big(i^{-1}G(w)\big)^{\sig_2(1-\eps_2)}\big(i^{-1}G(1/w)\big)^{\sig_1(2\eps_1-1)(1-\eps_2)}\\&+a^{1-\eps_2}\big(i^{-1}G(w)\big)^{-\sig_2\eps_2}\big(i^{-1}G(1/w)\big)^{\sig_1\eps_2(2\eps_1-1)}.\nonumber
\end{align}

If we rotate the graph $G$ by 45 degrees anti-clockwise, $\alpha$ represents the difference in the vertical direction between the dimers $e_j$ and $e_i$, and $\beta$ represents their horizontal difference. Next we have a lemma which is used in  the substitution of \eqref{temp0392} into \eqref{K_aalt}, and a lemma rewriting  the resulting integral in a form ready for asymptotic analysis.
\begin{lemma}\label{formk1l1k2l2} Consider $(x(j),y(i))\in W_{\eps_1}\times B_{\eps_2}$ as in \eqref{temp0392}, we have
\begin{align}
&|k_1|=\frac{|\alpha+\beta|}{2\eps}+\sig_1(2\eps_1-1)(1-\eps_2), &&|\ell_1|=\frac{|\alpha-\beta|}{2\eps}+\sig_2(1-\eps_2),\\
&|k_2|=\frac{|\alpha+\beta|}{2\eps}+\sig_1\eps_2(2\eps_1-1), &&|\ell_2|=\frac{|\alpha-\beta|}{2\eps}-\sig_2\eps_2.
\end{align}
\begin{proof}
We just show the formula for $|k_1|$, we have
\begin{align}
|k_1|&=|\frac{x_2(j)-y_2(i)-1}{2}+h(\eps_1,\eps_2)|\\
&=|\frac{(1+\alpha_j\eps^{-1}+\beta_j\eps^{-1}+2\eps_1-1)-(1+\alpha_i\eps^{-1}+\beta_i\eps^{-1})-1}{2}+h(\eps_1,\eps_2)|\nonumber\\
&=|\frac{(\alpha+\beta)}{2\eps}+\eps_1-1+h(\eps_1,\eps_2)|\nonumber
\end{align} 
and $\eps_1-1+h(\eps_1,\eps_2)=(2\eps_1-1)(1-\eps_2)$ by \eqref{HHH}. The formula for $|k_1|$ follows by the definition of $\sig_1$ in \eqref{sig1sig2} immediately for the case $\alpha=-\beta$ and for the case $\alpha\neq -\beta$ use the identity
\begin{align}
|x+y|=|y|+\text{sign}(y) \cdot x, \quad \text{ when } |x|\leq |y|
\end{align}
to get
\begin{align}
|k_1|=\frac{|\alpha+\beta|}{2\eps}+\text{sign}(\alpha+\beta)(2\eps_1-1)(1-\eps_2), \quad \text{ when } \frac{|\alpha+\beta|}{2\eps}\geq |(2\eps_1-1)(1-\eps_2)|.
\label{temp23d4ef}
\end{align}
Since $\eps_1,\eps_2\in \{0,1\}$ and $\frac{\alpha+\beta}{2\eps}\in \Z\setminus\{0\}$, one sees the inequality in \eqref{temp23d4ef} is satisfied for $\alpha\neq -\beta$.
\end{proof}
\end{lemma}
\begin{lemma}\label{preasymplem}
Consider $(x(j),y(i))\in W_{\eps_1}\times B_{\eps_2}$ as in \eqref{temp0392}, we have the integral representation
\begin{align}
\K_a^{-1}(x(j),y(i))&=-\frac{i^{1+h(\eps_1,\eps_2)}}{(1+a^2)\pi }\mcR\Big [\int_0^{\pi/2}\frac{dt}{|e^{2it}-2c|}\big(i^{-1}G(ie^{it})\big)^{\frac{|\alpha-\beta|}{2\eps}}\big(i^{-1}G(-ie^{-it}))^{\frac{|\alpha+\beta|}{2\eps}} g_{\eps_1,\eps_2}^{\sig_1,\sig_2}(ie^{it})\Big]\label{preasymp}
\end{align}
where $\mcR(z)$ denotes the real part of $z\in \C$.
\begin{proof}
Substitute the formulas in lemma \ref{formk1l1k2l2} into \eqref{K_aalt} and parametrise $w=e^{i\theta}$ to obtain
\begin{align}
\K_a^{-1}(x(j),y(i))&=-\frac{i^{1+h(\eps_1,\eps_2)}}{4(1+a^2)\pi i}\int_{\Gamma_1}\frac{dw}{w}\frac{\big(i^{-1}G(w)\big)^{|\alpha-\beta|/2\eps}\big(i^{-1}G(1/w))^{|\alpha+\beta|/2\eps}}{\sqrt{w^2+2c}\sqrt{1/w^2+2c}} g_{\eps_1,\eps_2}^{\sig_1,\sig_2}(w)\\
&=-\frac{i^{1+h(\eps_1,\eps_2)}}{4(1+a^2)\pi }\int_{-\pi}^\pi d\theta\frac{\big(i^{-1}G(e^{i\theta})\big)^{|\alpha-\beta|/2\eps}\big(i^{-1}G(e^{-i\theta}))^{|\alpha+\beta|/2\eps}}{|e^{2i\theta}+2c|} g_{\eps_1,\eps_2}^{\sig_1,\sig_2}(e^{i\theta}).
\end{align}
Split the above integral up into the sum over $(-\pi,0)$ and $(0,\pi)$. Make the change of variables $\theta\rarrow -\theta$ just over the region $\theta\in(-\pi,0)$. We have the symmetry
\begin{align}
\overline{g_{\eps_1,\eps_2}^{\sig_1,\sig_2}(e^{-i\theta})}=& \  a^{\eps_2}\big(i^{-1}G(e^{i\theta})\big)^{\sig_2(1-\eps_2)}\big(i^{-1}G(e^{-i\theta})\big)^{\sig_1(2\eps_1-1)(1-\eps_2)}(-1)^{\sig_2(1-\eps_2)+\sig_1(2\eps_1-1)(1-\eps_2)}\\&+a^{1-\eps_2}\big(i^{-1}G(e^{i\theta})\big)^{-\sig_2\eps_2}\big(i^{-1}G(e^{-i\theta})\big)^{\sig_1\eps_2(2\eps_1-1)}(-1)^{-\sig_2\eps_2+\sig_1\eps_2(2\eps_1-1)}\nonumber\\
=&\ g_{\eps_1,\eps_2}^{\sig_1,\sig_2}(e^{i\theta})\nonumber
\end{align}
where we used 
\begin{align}
i^{-1}G(e^{-i\theta})=-\overline{(i^{-1}G(e^{i\theta}))}
\end{align} in the first line and in the second line we used the facts $\sig_1,\sig_2\in \{-1,1\}$, $(1-\eps_2)+(2\eps_1-1)(1-\eps_2) \text{ mod }2=0$ and $-\eps_2+\eps_2(2\eps_1-1) \text{ mod }2 =0$.
We now have
\begin{align}
\K_a^{-1}(x(j),y(i))&=-\frac{i^{1+h(\eps_1,\eps_2)}}{2(1+a^2)\pi }\mcR\Big[\int_{0}^\pi d\theta\frac{\big(i^{-1}G(e^{i\theta})\big)^{|\alpha-\beta|/2\eps}\big(i^{-1}G(e^{-i\theta}))^{|\alpha+\beta|/2\eps}}{|e^{2i\theta}+2c|} g_{\eps_1,\eps_2}^{\sig_1,\sig_2}(e^{i\theta})\Big].
\end{align}
Now we use the change of variable $\theta=t+\pi/2$, $t\in (-\pi/2,\pi/2)$, followed by the symmetries
\begin{align}
i^{-1}G(ie^{-it})=\overline{(i^{-1}G(ie^{it}))},&& g_{\eps_1,\eps_2}^{\sig_1,\sig_2}(ie^{it})=\overline{g_{\eps_1,\eps_2}^{\sig_1,\sig_2}(ie^{-it})}
\end{align}
which yields the integral formula \eqref{preasymp}.
\end{proof}
\end{lemma}

\subsection{Asymptotics analysis}
In this section we perform an asymptotic analysis on the expression for $\K_a^{-1}$ found in Lemma \ref{preasymplem}.
The analysis is akin to Laplace's method, but instead of approximating a saddle-point function $f(x)$ near its critical point $x_0$ by a quadratic $f(x_0)+f''(x_0)x^2/2$ (a saddle) we will approximate the associated "saddle-point" function by a complex-valued function involving square roots, see Lemma \ref{'saddleptfn'}.
Define 
\begin{align}
r_1=\frac{|\alpha-\beta|}{2},&&r_2=\frac{|\alpha+\beta|}{2}
\end{align}
and the "saddle-point" function
\begin{align}
f_{r_1,r_2}(t)=r_1\log(i^{-1}G(ie^{it}))+r_2\log(i^{-1}G(-ie^{-it})).
\end{align}
The formula \eqref{preasymp} then reads
\begin{align}
\K_a^{-1}(x(j),y(i))&=-\frac{i^{1+h(\eps_1,\eps_2)}}{(1+a^2)\pi }\mcR\Big [\int_0^{\pi/2}\frac{dt}{|e^{2it}-2c|}\exp(\frac{1}{\eps}f_{r_1,r_2}(t)) g_{\eps_1,\eps_2}^{\sig_1,\sig_2}(ie^{it})\Big].\label{preasymp12}
\end{align}

We perform an asymptotic analysis on the integral which appears under the real-part sign in \eqref{preasymp12}. 
For $\eps_1,\eps_2\in\{0,1\}$, $\sig_1,\sig_2\in \{-1,1\}$
and for $t>0$, define
\begin{align}
\tilde{g}_{\eps_1,\eps_2}^{\sig_1,\sig_2}(t)=-1-\frac{(-1)^{\eps_1}\sig_1}{\sqrt{2}}\sqrt{1-4it}+\frac{(-1)^{\eps_2}\sig_2}{\sqrt{2}}\sqrt{1+4it}
\end{align}
where the square roots are principal branch square roots.

We require some facts about our "saddle-point" function $f_{r_1,r_2}$.
\begin{lemma}\label{descent}
If $0< t_1<t_2 <  \pi/2$ and one of $r_1,r_2$ is greater than zero, then  $\mcR[f_{r_1,r_2}(t_2)]<\mcR[f_{r_1,r_2}(t_1)]$, that is, $f_{r_1,r_2}$ is a strictly decreasing function.
\begin{proof}
Follows from Lemma 10 in \cite{J.M} by the symmetry $G(w)=\overline{G(\overline{w})}$.
\end{proof}
\end{lemma}
\begin{lemma}\label{'saddleptfn'}
Let $a=1-\eps>0$, $\eps>0$, $d_\eps\sim \eps^{-1/2}$ and $0<t\leq d_\eps\pi/2$, there is a bounded function $R_1(t,\eps)$ such that
\begin{align}
\frac{1}{\eps}f_{r_1,r_2}(t\eps^2)=-i\pi r_2/\eps-\frac{1}{\sqrt{2}}\big(r_1\sqrt{1+4it}+r_2\sqrt{1-4it}\big)+R_1(t,\eps).
\end{align}
Furthermore, we have a constant $C>0$ such that
\begin{align}
|R_1(t,\eps)|<C\eps^{1/2}
\end{align}  for $0<t\leq d_\eps\pi/2$, $\eps$ sufficiently small and $(r_1,r_2)\in \R^2_{\geq 0}\setminus\{0\}$ in compact sets.
\begin{proof}
We have
\begin{align}
&f_{r_1,r_2}(t\eps^2)\label{temp0432}\\&=r_1\log(i^{-1}\frac{1}{\sqrt{2c}}(ie^{it\eps^2}-\sqrt{(ie^{it\eps^2})^2+2c}))+r_2\log(i^{-1}\frac{1}{\sqrt{2c}}(-ie^{-it\eps^2}-\sqrt{(-ie^{-it\eps^2})^2+2c})).\nonumber
\end{align}
We focus on the first term on the right-hand side above.
By Taylor's theorem we have bounded functions $R_2(\eps)$ for $0<\eps<1$ and $R_3(t\eps^2)$ for $0<t\eps^2<\pi/2$ such that
\begin{align}
\sqrt{2c}=1-\eps^2/4+R_2(\eps)\eps^3&& i\big(ie^{it\eps^2}\big)=-1-it\eps^2-t^2\eps^4R_3(t\eps^2).
\end{align}
Recall that $\sqrt{w^2+2c}=i\sqrt{-\sqrt{2c}-iw}\sqrt{\sqrt{2c}-iw}$ where the square roots on the right are principal branch square roots. Inserting the above expansions into  the principal branch square roots, another application of Taylor's  theorem yields two functions $R_5(t,\eps),R_6(t,\eps)$ both bounded for $0<\eps<1$, $0<t\eps^2<\pi/2$ such that
\begin{align}
\sqrt{(ie^{it\eps^2})^2+2c}=& \ i\sqrt{-\sqrt{2c}-i(ie^{it\eps^2})}\cdot \sqrt{\sqrt{2c}-i(ie^{it\eps^2})}
\\=& \ i\eps\sqrt{it+1/4}(1+R_5(t,\eps)[R_3(\eps)\eps+t^2\eps^2R_4(t\eps^2)])\nonumber\\
&\times\sqrt{2}(1+R_6(t,\eps)[\eps^2/4-R_3(\eps)\eps^3+it\eps^2+t^2\eps^4R_4(t\eps^2)]).\nonumber
\end{align}
Expanding out the above brackets we can see there is a bounded function $R_7(t,\eps)$ where $0<\eps<1$, $0<t\leq d_\eps \pi/2$ such that
\begin{align}
\sqrt{(ie^{it\eps^2})^2+2c}=\frac{i\eps}{\sqrt{2}}\sqrt{1+4it}+R_7(t,\eps)\eps^2.
\end{align}
Now we have $R_8(t\eps^2)$ bounded on $0<t\eps^2<\pi/2$ such that $e^{it\eps^2}=1+R_8(t\eps^2)t\eps^2$ and $R_9(\eps)$ bounded on $0<\eps<1$ such that $1/\sqrt{2c}=1+R_9(\eps)\eps^2$, so we get
\begin{align}
&r_1\log(i^{-1}\frac{1}{\sqrt{2c}}(ie^{it\eps^2}-\sqrt{(ie^{it\eps^2})^2+2c}))\\&=r_1\log(\frac{1}{\sqrt{2c}}(e^{it\eps^2}-\frac{\eps}{\sqrt{2}}\sqrt{1+4it}+i^{-1}R_7(t,\eps)\eps^2))\nonumber\\
&=r_1\log(1+R_9(\eps)\eps^2)+r_1\log(1-\frac{\eps}{\sqrt{2}}\sqrt{1+4it}+R_8(t\eps^2)t\eps^2+i^{-1}R_7(t,\eps)\eps^2).\nonumber
\end{align}
By Taylor's theorem applied to the logarithms in the previous line, we see we get a bounded function $R_{10}(t,\eps)$ where $0<t\leq d_\eps\pi/2$ and $\eps>0$ is small, such that
\begin{align}
r_1\log(i^{-1}\frac{1}{\sqrt{2c}}(ie^{it\eps^2}-\sqrt{(ie^{it\eps^2})^2+2c}))=-r_1\frac{\eps}{\sqrt{2}}\sqrt{1+4it}+R_{10}(t,\eps)\eps^{3/2}.\label{temp5190}
\end{align}
Next we note that 
\begin{align}
&r_2\log(i^{-1}\frac{1}{\sqrt{2c}}(-ie^{-it\eps^2}-\sqrt{(-ie^{-it\eps^2})^2+2c}))\\& \ =-i\pi r_2 +r_2\log(-i^{-1}\frac{1}{\sqrt{2c}}(-ie^{-it\eps^2}-\sqrt{(-ie^{-it\eps^2})^2+2c})).\nonumber
\end{align}
Since the complex conjugate of \eqref{temp5190} is the second factor in the previous expression, we see that the lemma holds with $R_1(t,\eps)=(R_{10}(t,\eps)+\overline{R_{10}(t,\eps)})\eps^{1/2}$.
\end{proof}
\end{lemma}
\begin{lemma}\label{geps1eps2lim}
Let $0<t<d_\eps\pi/2$, $\eps>0$, $d_\eps\sim \eps^{-1/2}$, then there is bounded function $R(t,\eps)$ such that
\begin{align}
\frac{1}{\eps}g_{\eps_1,\eps_2}^{\sig_1,\sig_2}(ie^{it\eps^2})=\tilde{g}_{\eps_1,\eps_2}^{\sig_1,\sig_2}(t)+R(t,\eps),
\end{align}
and there is a constant $C>0$ such that
\begin{align}
|R(t,\eps)|\leq C\eps^{1/2}
\end{align}
uniformly in $t$.
\begin{proof}
From the definition of $g$ in \eqref{gfn}, 
\begin{align}\label{temp041}
g_{\eps_1,\eps_2}^{\sig_1,\sig_2}(ie^{it\eps^2})=& \ a^{\eps_2}\big(i^{-1}G(ie^{it\eps^2})\big)^{\sig_2(1-\eps_2)}\big(-i^{-1}G(-ie^{-it\eps^2})\big)^{\sig_1(2\eps_1-1)(1-\eps_2)}(-1)^{\sig_1(2\eps_1-1)(1-\eps_2)}\\&+a^{1-\eps_2}\big(i^{-1}G(ie^{it\eps^2})\big)^{-\sig_2\eps_2}\big(-i^{-1}G(-ie^{-it\eps^2})\big)^{\sig_1\eps_2(2\eps_1-1)}(-1)^{\sig_1\eps_2(2\eps_1-1)}.\nonumber\\
=& \ h_{\eps_1,\eps_2}^{\sig_1,\sig_2}(ie^{it\eps^2})+h_{\eps_1,1-\eps_2}^{\sig_1,\sig_2}(ie^{it\eps^2})\label{tem20}
\end{align}
where
\begin{align}
h_{\eps_1,\eps_2}^{\sig_1,\sig_2}(ie^{it\eps^2}):=a^{\eps_2}\big(i^{-1}G(ie^{it\eps^2})\big)^{\sig_2(1-\eps_2)}\big(-i^{-1}G(-ie^{-it\eps^2})\big)^{\sig_1(2\eps_1-1)(1-\eps_2)}(-1)^{\sig_1(2\eps_1-1)(1-\eps_2)}.\label{tmph}
\end{align}
We will focus on a Taylor expansion for $h_{\eps_1,\eps_2}^{\sig_1,\sig_2}$.
From the Taylor expansions performed in lemma \ref{'saddleptfn'}, we have a bounded function $R_{11}(t,\eps)$ on $0<t<d_\eps \pi/2$, $\eps>0$ such that
\begin{align}
i^{-1}G(ie^{it\eps^2})=1-\frac{\eps}{\sqrt{2}}\sqrt{1+4it}+R_{11}(t,\eps)\eps^{3/2},\nonumber\\
-i^{-1}G(-ie^{-it\eps^2})=1-\frac{\eps}{\sqrt{2}}\sqrt{1-4it}+\overline{R_{11}(t,\eps)}\eps^{3/2}.\nonumber
\end{align}
We insert these formulae into \eqref{tmph}  and use the following identity: for $\sigma\in \{-1,0,1\}$, $x\in\C$, $\eps$ small enough, we have 
\begin{align}
\big(1-\eps x\big)^{\sigma}=1-\sigma\eps x+\sigma\eps^2x^2R^{\sigma}(\eps,x)
\end{align}
where $R^0=R^1=0$, and $R^{-1}$ is bounded function given by Taylor's theorem.
After we use this identity we can see we have a bounded function  $R_{12}(t,\eps)$ on $0<t<d_\eps\pi/2$, $\eps>0$ small enough, such that 
\begin{align}
&h_{\eps_1,\eps_2}^{\sig_1,\sig_2}(ie^{it\eps^2})\nonumber\\&=\big(1-\eps\big)^{\eps_2}\big(1-\sig_2(1-\eps_2)\frac{\eps}{\sqrt{2}}\sqrt{1+4it}\big)\big(1-\sig_1(2\eps_1-1)(1-\eps_2)\frac{\eps}{\sqrt{2}}\sqrt{1-4it}\big)(-1)^{\sig_1(2\eps_1-1)(1-\eps_2)}\nonumber\\
& \  \ \ + R_{12}(t,\eps)\eps^{3/2}.\nonumber
\end{align}
Next observe that $(1-\eps)^{\eps_2}=1-\eps_2\eps$, $(-1)^{\sig_1(2\eps_1-1)(1-\eps_2)}=(-1)^{1-\eps_2}$ and the bound $|\sqrt{1+4it}|\leq C\eps^{-1/2}$. This gives a bounded function $R_{13}(t,\eps)$ such that
\begin{align}
&h_{\eps_1,\eps_2}^{\sig_1,\sig_2}(ie^{it\eps^2})\nonumber\\&=\big(1-\eps_2\eps-\sig_2(1-\eps_2)\frac{\eps}{\sqrt{2}}\sqrt{1+4it}-\sig_1(2\eps_1-1)(1-\eps_2)\frac{\eps}{\sqrt{2}}\sqrt{1-4it}\big)(-1)^{(1-\eps_2)}\nonumber\\
& \  \ \ + R_{13}(t,\eps)\eps^{3/2}.\nonumber
\end{align}
Next observe that $(-1)^{1-\eps_2}(1-\eps_2)=-(1-\eps_2)$, $(-1)^{1-\eps_2}\eps_2=\eps_2$ so that
\begin{align}
&h_{\eps_1,\eps_2}^{\sig_1,\sig_2}(ie^{it\eps^2})\nonumber\\&=(-1)^{1-\eps_2}-\eps_2\eps+\sig_2(1-\eps_2)\frac{\eps}{\sqrt{2}}\sqrt{1+4it}+\sig_1(2\eps_1-1)(1-\eps_2)\frac{\eps}{\sqrt{2}}\sqrt{1-4it}\nonumber\\
& \  \ \ + R_{13}(t,\eps)\eps^{3/2}.\nonumber
\end{align}
The lemma now follows by \eqref{tem20} and the identity $(-1)^{1-\eps_2}+(-1)^{-\eps_2}=0$.
\end{proof}
\end{lemma}
\begin{proposition}\label{prop105}
We have the limit
\begin{align}
&\frac{(-1)^{r_2/\eps}}{\eps}\int_0^{\pi/2}\frac{dt}{|e^{2it}-2c|}\exp(\frac{1}{\eps}f_{r_1,r_2}(t)) g_{\eps_1,\eps_2}^{\sig_1,\sig_2}(ie^{it})\label{temp4rfz}\\ & \ \rarrow 2\int_0^\infty \frac{dt}{\sqrt{1+16t^2}}\exp\big(-\frac{1}{\sqrt{2}}\big(r_1\sqrt{1+4it}+r_2\sqrt{1-4it}\big)\big)\tilde{g}_{\eps_1,\eps_2}^{\sig_1,\sig_2}(t)\nonumber
\end{align}
as $\eps\rarrow 0$ uniformly for $(r_1,r_2)$ in a compact subset of  $\R_{\geq 0}^2\setminus\{0\}$.
\begin{proof}
Make the change of variables $t\rarrow \eps^2t$ and write
\begin{align}
(-1)^{r_2/\eps}\eps\int_0^{\pi/(2\eps^2)}\frac{dt}{|e^{2it\eps^2}-2c|}\exp(\frac{1}{\eps}f_{r_1,r_2}(t\eps^2)) g_{\eps_1,\eps_2}^{\sig_1,\sig_2}(ie^{it\eps^2}).\label{temp23eds}
\end{align}
Compute $|e^{2i \eps^2t}-2c|^2=(1-2c)^2+8c\sin^2\eps^2t=\eps^4/4+4\sin^2\eps^2t+\eps^5R_5(\eps)$ for a bounded function $R_5(\eps,t)$ on $\eps>0,t>0$. Let $d_\eps\sim \eps^{-1/2}$, for $0<t\leq d_\eps\pi/2$, $\frac{\sin^2\eps^2t}{\eps^4}=t^2+\eps^4t^4R_6(\eps^2t)=t^2+\eps^2R_7(\eps,t)$ where $R_6$ is a bounded function that comes from Taylor's theorem and $R_7$ is a bounded function since $\eps^4t^4\leq \eps^2\pi/2$. Another application of Taylor's theorem gives a bounded function $R_8(t,\eps)$ defined on $0<\eps<1,0<t\leq d_\eps\pi/2$ such that
\begin{align}
\frac{\eps^2}{|e^{2i\eps^2t}-2c|}=\frac{2}{1+16t^2}+\eps R_8(\eps,t).\label{temp4rzx}
\end{align}  
One can also show that there is a $C$ such that $|e^{2i\eps^2t}-2c|^{-1}\leq C\eps^{-2}$ for all $t>0$ and $0<\eps<1$. 
The main contribution to the asymptotics of \eqref{temp23eds} comes from the section of the integral over $(0,d_\eps \pi/2)$: 
\begin{align}
(-1)^{r_2/\eps}\eps\int_0^{d_\eps \pi/2}\frac{dt}{|e^{2it\eps^2}-2c|}\exp(\frac{1}{\eps}f_{r_1,r_2}(t\eps^2)) g_{\eps_1,\eps_2}^{\sig_1,\sig_2}(ie^{it\eps^2}).\label{temp23edf}
\end{align}
We now make a collection of successive approximations to the integral in \eqref{temp23edf} and show that the errors tend to zero as $\eps\rarrow 0$.
The first approximation is
\begin{align}
\eps\int_0^{d_\eps\pi/2}\frac{dt}{|e^{2it\eps^2}-2c|}\exp(-\frac{1}{\sqrt{2}}(r_1\sqrt{1+4it}+r_2\sqrt{1-4it})) g_{\eps_1,\eps_2}^{\sig_1,\sig_2}(ie^{it\eps^2}).\label{temp543df}
\end{align}
The modulus of the difference between \eqref{temp23edf} and \eqref{temp543df} (i.e. the error) is bounded above by
\begin{align}
\eps \int_0^{d_\eps\pi/2}\frac{dt}{|e^{2it\eps^2}-2c|}|(-1)^{r_2/\eps}e^{\frac{1}{\eps}f_{r_1,r_2}(t\eps^2)}-e^{-\frac{1}{\sqrt{2}}(r_1\sqrt{1+4it}-r_2\sqrt{1-4it})}||g_{\eps_1,\eps_2}^{\sig_1,\sig_2}(ie^{it\eps^2})|.\label{temp24dxz}
\end{align}
By Lemma \ref{'saddleptfn'} we have
\begin{align}
&|(-1)^{r_2/\eps}e^{\frac{1}{\eps}f_{r_1,r_2}(t\eps^2)}-e^{-\frac{1}{\sqrt{2}}(r_1\sqrt{1+4it}-r_2\sqrt{1-4it})}|\\
&\quad\leq e^{\mcR[\frac{1}{\sqrt{2}}(r_1\sqrt{1+4it}+r_2\sqrt{1-4it})]}|R_1(t,\eps)|e^{|R_1(t,\eps)}|\\
&\quad\leq C_1\eps^{1/2}e^{C_1\eps^{1/2}-\mcR[\frac{1}{\sqrt{2}}(r_1\sqrt{1+4it}+r_2\sqrt{1-4it})]}\nonumber
\end{align}
which, together with Lemma \ref{geps1eps2lim}, gives the upper bound on \eqref{temp24dxz}
\begin{align}
\eps^2\int_0^{d_\eps\pi/2}\frac{dt}{|e^{2it\eps^2}-2c|}C_1\eps^{1/2}e^{C_1\eps^{1/2}-\mcR[\frac{1}{\sqrt{2}}(r_1\sqrt{1+4it}+r_2\sqrt{1-4it})]}(|g_{\eps_1,\eps_2}^{\sigma_1,\sigma_2}(t)|+|C_2\eps^{1/2}|)\nonumber\\
\leq C\eps^{1/2}C_1\int_0^{d_\eps\pi/2}dte^{C_1\eps^{1/2}-\mcR[\frac{1}{\sqrt{2}}(r_1\sqrt{1+4it}+r_2\sqrt{1-4it})]}(|g_{\eps_1,\eps_2}^{\sigma_1,\sigma_2}(t)|+|C_2\eps^{1/2}|).\label{temp23edgc}
\end{align}
Now we note that $\mcR[\sqrt{1+4it}]=\mcR[\sqrt{1-4it}]$ and that there is an $r^*>0$ such that $r_1+r_2\geq r^*$.
From this we see \eqref{temp23edgc} is bounded above by
\begin{align}
C\eps^{1/2}C_1e^{C_1\eps^{1/2}}\int_0^{\infty}dte^{-\mcR[\frac{1}{\sqrt{2}}r^*\sqrt{1+4it}]}(|g_{\eps_1,\eps_2}^{\sigma_1,\sigma_2}(t)|+|C_2\eps^{1/2}|)
\end{align}
which clearly tends to zero as $\eps\rarrow 0$ as one sees that the integrand is integrable.
We now approximate \eqref{temp543df} by 
\begin{align}
\eps\int_0^{d_\eps\pi/2}\frac{dt}{|e^{2it\eps^2}-2c|}\exp(-\frac{1}{\sqrt{2}}(r_1\sqrt{1+4it}+r_2\sqrt{1-4it})) \eps\tilde{g}_{\eps_1,\eps_2}^{\sig_1,\sig_2}(t)\label{temp5rd}
\end{align}
Indeed, by lemma \ref{geps1eps2lim}, the modulus of the difference (the error) between \eqref{temp5rd} and \eqref{temp543df} is bounded above by
\begin{align}
&C_2\eps^{5/2}\int_0^{d_\eps\pi/2}\frac{dt}{|e^{2it\eps^2}-2c|}\exp(-\frac{r^*}{\sqrt{2}}\mcR[\sqrt{1+4it}])\\&\leq CC_2\eps^{1/2}\int_0^{\infty}dt\exp(-\frac{r^*}{\sqrt{2}}\mcR[\sqrt{1+4it}])
\end{align}
which tends to zero as $\eps\rarrow 0$. Finally we approximate \eqref{temp5rd} by
\begin{align}
2\int_0^\infty \frac{dt}{\sqrt{1+16t^2}}\exp\big(-\frac{1}{\sqrt{2}}\big(r_1\sqrt{1+4it}+r_2\sqrt{1-4it}\big)\big)\tilde{g}_{\eps_1,\eps_2}^{\sig_1,\sig_2}(t)
\end{align}
which by \eqref{temp4rzx} has an error of
\begin{align}
C_3\eps \int_0^{\infty}dt\exp(-\frac{r^*}{\sqrt{2}}\mcR[\sqrt{1+4it}])\rarrow 0.
\end{align}
All that remains is to show the section of the integral in \eqref{temp23eds} over $(d_\eps\pi/2,\pi/(2\eps^2))$ tends to zero. It is straightforward to see  $g_{\eps_1,\eps_2}^{\sig_1,\sig_2}(ie^{it\eps^2})$ is bounded uniformly for all $0<t<\pi/(2\eps^2)$, $0<\eps<1$. Now
\begin{align}
&\big|\eps\int_{d_\eps\pi/2}^{\pi/(2\eps^2)}\frac{dt}{|e^{2it\eps^2}-2c|}\exp(\frac{1}{\eps}f_{r_1,r_2}(t\eps^2)) g_{\eps_1,\eps_2}^{\sig_1,\sig_2}(ie^{it\eps^2})\big|\label{temp310}\\
& \ \leq C\eps^{-1}\exp\big(\max_{t\in (d_\eps\pi/2,\pi/(2\eps^2))}\mcR[\frac{1}{\eps}f_{r_1,r_2}(t\eps^2)]\big)(\pi/(2\eps^2)-d_\eps\pi/2).\nonumber
\end{align}
We know from Lemma \ref{descent} that $\mcR[\frac{1}{\eps}f_{r_1,r_2}(t\eps^2)]$ is decreasing in $t>0$, hence by Lemma \ref{'saddleptfn'},
\begin{align}
&\max_{t\in (d_\eps\pi/2,\pi/(2\eps^2))}\mcR[\frac{1}{\eps}f_{r_1,r_2}(t\eps^2)]\leq \mcR[\frac{1}{\eps}f_{r_1,r_2}(d_\eps\pi\eps^2/2)]\\& \ \leq -\mcR[\frac{1}{\sqrt{2}}\big(r_1\sqrt{1+4id_\eps\pi/2}+r_2\sqrt{1-4id_\eps\pi/2}\big)+R_1(d_\eps\pi/2,\eps)]\nonumber\\
& \ \leq -C'\eps^{-1/4}\nonumber
\end{align} for some $C'>0$. Hence the bound in \eqref{temp310} decays exponentially as $\eps\rarrow 0$. From our proof we see that the limit \eqref{temp4rfz} holds with an error term of order $\eps^{1/2}$.
\end{proof}
\end{proposition}

We are now ready to prove theorem \ref{invkastasymp}.
\begin{proof}[Proof of Theorem \ref{invkastasymp}]\label{proofinvKastasymp}
Assume first that $\lambda=1$. We recall the definitions of $\sig_1,\sig_2$ in \eqref{sig1sig2} and set $r_1=|\alpha-\beta|/2, r_2=|\alpha+\beta|/2$. We use the formula  for $\K_a^{-1}$ given by \eqref{preasymp12} and proposition \ref{prop105} to get 
\begin{align}
\frac{(-1)^{(\alpha+\beta)\eps^{-1}/2}}{\eps}\K_a^{-1}(x(j),y(i))= \frac{i^{1+h(\eps_1,\eps_2)}}{\pi}\mcR\int_0^\infty \frac{dt}{\sqrt{1+16t^2}}\exp\big(-\frac{1}{\sqrt{2}}\big(r_1\sqrt{1+4it}+r_2\sqrt{1-4it}\big)\big)\label{temp5ko9}\\
\Big(1+\frac{(-1)^{\eps_1}\sig_1}{\sqrt{2}}\sqrt{1-4it}-\frac{(-1)^{\eps_2}\sig_2}{\sqrt{2}}\sqrt{1+4it}\Big)+o(1)\nonumber
\end{align}
as $\eps\rarrow 0$, where $h(\eps_1,\eps_2)$ is defined in \eqref{HHH}. From Proposition \ref{prop64} in the Appendix below one can see the identity
\begin{align}
&\mcR\int_0^\infty \frac{dt}{\sqrt{1+16t^2}}\exp\big(-\frac{1}{\sqrt{2}}\big(r_1\sqrt{1+4it}+r_2\sqrt{1-4it}\big)\big)\big(\frac{(-1)^{\eps_1}\sig_1}{\sqrt{2}}\sqrt{1-4it}-\frac{(-1)^{\eps_2}\sig_2}{\sqrt{2}}\sqrt{1+4it}\big)\\ & \ =\frac{(-1)^{\eps_1}\sig_1r_2-(-1)^{\eps_2}\sig_2r_1}{2\sqrt{r_1^2+r_2^2}}K_1( \ \sqrt{r_1^2+r_2^2} \ ).\nonumber
\end{align}
When $\alpha\neq \beta$ and $\alpha\neq -\beta$, \eqref{sig1sig2} gives $\sig_1=\text{sign}(\alpha+\beta)$ and $\sig_2=\text{sign}(\alpha-\beta)$, in which case one can see
\begin{align}
(-1)^{\eps_1}\sig_1r_2-(-1)^{\eps_2}\sig_2r_1&= (-1)^{\eps_1}(\alpha+\beta)/2-(-1)^{\eps_2}(\alpha-\beta)/2\label{temp0563}\\
&=((-1)^{\eps_2}-(-1)^{\eps_1})\frac{\alpha}{2}-((-1)^{\eps_2}+(-1)^{\eps_1})\frac{\beta}{2}\nonumber
\\
&=(-1)^{\eps_2}(\ind_{\eps_1\neq \eps_2}\alpha-\ind_{\eps_1=\eps_2}\beta)\nonumber.
\end{align}
One can see that for the cases when either $\alpha=\beta$ or $\alpha=-\beta$, the identity \eqref{temp0563} still holds. We also have $\sqrt{r_1^2+r_2^2}=\sqrt{\alpha^2+\beta^2}/\sqrt{2}$. Hence proposition \ref{prop64}, \eqref{temp0563} and the limit \eqref{temp5ko9} give
\begin{align}
\label{tempr4ea0}&\frac{(-1)^{(\alpha+\beta)\eps^{-1}/2}}{\eps}\K_a^{-1}(x(j),y(i))=\frac{i^{1+h(\eps_1,\eps_2)}}{2\pi}\Big(K_0(\sqrt{\alpha^2+\beta^2}/\sqrt{2})\\&-(-1)^{\eps_2}(\ind_{\eps_1\neq \eps_2}\sqrt{2}\alpha-\ind_{\eps_1=\eps_2}\sqrt{2}\beta)\frac{K_1(\sqrt{\alpha^2+\beta^2}/\sqrt{2})}{\sqrt{\alpha^2+\beta^2}}\Big)+o(1)\nonumber
\end{align}
as $\eps>0$ tends to zero. Now \eqref{tempr4ea0} gives the four limits in the theorem statement for the case $\lambda=1$. If we instead assume $\lambda>0$, then set $a=1-\lambda \eps=1-\eps'$ where $\eps'=\lambda \eps$. Rescaling the coordinates in \eqref{temp039112} by substituting $\eps=\eps'/\lambda$, one can then use the limits in the case $\lambda=1$ to prove the case $\lambda >0$. 
\end{proof}

\appendix
\section{Identities for Bessel functions of the second kind}
In this Appendix we prove some identities for Bessel functions of the second kind, $K_\nu(z)$, $\nu=0,1$. These are used to rewrite the limiting integral expression for $\K_a^{-1}$ in terms of $K_\nu$. We recall the following well-known integral representation of $K_\nu$, \cite{DLMF},
\begin{align}
K_\nu(z)=\frac{\pi^{1/2}(\frac{1}{2}z)^v}{\Gamma(\nu+1/2)}\int_1^\infty e^{-zt}(t^2-1)^{\nu-1/2}dt, \quad z>0, \ \nu=0,1,2... \ .
\end{align}
We first give alternative integral representations of $K_0$ and $K_1$.
\begin{lemma}\label{lemmaintreps}
Let $(r_1,r_2)\in \R^2_{\geq 0}\setminus\{0\}$, then
\begin{align}
K_0\big( \ \sqrt{r_1^2+r_2^2} \ \big)&=\int_1^\infty\frac{du}{\sqrt{u^2-1}}e^{-\frac{r_1+r_2}{\sqrt{2}}u}\cos\big(\frac{r_1-r_2}{\sqrt{2}}\sqrt{u^2-1}\big),\label{temp5290}\\
\frac{r_1+r_2}{\sqrt{r_1^2+r_2^2}}K_1\big( \ \sqrt{r_1^2+r_2^2} \ \big)&=\int_1^\infty\frac{\sqrt{2}udu}{\sqrt{u^2-1}}e^{-\frac{r_1+r_2}{\sqrt{2}}u}\cos\big(\frac{r_1-r_2}{\sqrt{2}}\sqrt{u^2-1}\big),\label{temp9025}\\
\frac{r_1-r_2}{\sqrt{r_1^2+r_2^2}}K_1\big( \ \sqrt{r_1^2+r_2^2} \ \big)&=\int_1^\infty \sqrt{2}due^{-\frac{r_1+r_2}{\sqrt{2}}u}\sin\big(\frac{r_1-r_2}{\sqrt{2}}\sqrt{u^2-1}\big).\label{temp59022}
\end{align}
\begin{proof}
 Set $z=\sqrt{r_1^2+r_2^2}$ so that $r_1=z\cos\theta$, $r_2=z\sin\theta$, $\theta\in[0,\pi/2]$. We have $\frac{r_1+r_2}{\sqrt{2}}=z\cos(\theta-\pi/4)$ and $\frac{r_1-r_2}{\sqrt{2}}=z\cos(\theta+\pi/4)$. The right-hand side of \eqref{temp5290} is
\begin{align}
 \mcR\int_1^\infty\frac{du}{\sqrt{u^2-1}}e^{-z(\cos(\theta-\pi/4)u-i\cos(\theta+\pi/4)\sqrt{u^2-1})}.\label{temp978}
\end{align}
Let $y=\cos(\theta-\pi/4)u-i\cos(\theta+\pi/4)\sqrt{u^2-1}$, one can show that 
\begin{align}
u&=\sqrt{1-y^2}\cos(\theta+\pi/4)+y\cos(\theta-\pi/4),\\
\frac{du}{dy}&=-\frac{y}{\sqrt{1-y^2}}\cos(\theta+\pi/4)+\cos(\theta-\pi/4).\label{temp0609}
\end{align}
Let $\gamma_\theta\subset \C$ be the image of the interval $[1,\infty)$ under the map $u\mapsto y$. One can see that $\gamma_\theta$ is a curve starting at the point $\cos(\theta-\pi/4)\in\C$ and travelling to infinity with increasing real part. Substitution in \eqref{temp978} gives \eqref{temp978} equal to
\begin{align}
\mcR\int_{\gamma_\theta}dye^{-zy}\frac{du}{dy}\frac{1}{\sqrt{u^2-1}}.\label{temp920t5}
\end{align}
We have the friendly algebraic fact 
\begin{align}
\Big(\frac{du}{dy}\frac{1}{\sqrt{u^2-1}}\Big)^2=\frac{1}{y^2-1},
\end{align}
and we take the square root of this equation and insert it into \eqref{temp920t5}. We then see we have an analytic integrand and deform $\gamma_\theta$ to the straight line $(\cos(\theta-\pi/4),\infty)\subset \C$ to obtain
\begin{align}
&\mcR\int_{(\cos(\theta-\pi/4),\infty)}dye^{-zy}\frac{1}{\sqrt{y^2-1}}\label{tepm32}\\
& = \int_{1}^\infty dye^{-zy}\frac{1}{\sqrt{y^2-1}}\nonumber\\
&  =K_0(z).\nonumber
\end{align}
The first equality in \eqref{tepm32} follows since the section of the integral over $(\cos(\theta-\pi/4),1)$ is purely imaginary.

Similarly, consider the right-hand side of \eqref{temp9025} as
\begin{align}
 \mcR\int_1^\infty\frac{\sqrt{2}udu}{\sqrt{u^2-1}}e^{-z(\cos(\theta-\pi/4)u-i\cos(\theta+\pi/4)\sqrt{u^2-1})}.\label{temp1240}
\end{align}
Making the substitution $u\rarrow y$ in \eqref{temp1240}, similarly to above one obtains
\begin{align}
&\sqrt{2}\mcR\int_{\gamma_\theta}dy\frac{y\cos(\theta-\pi/4)+\sqrt{1-y^2}\cos(\theta+\pi/4)}{\sqrt{y^2-1}}e^{-zy}\label{temp420999}\\
&  \ =\sqrt{2}\mcR\int_{\cos(\theta-\pi/4)}^\infty dy\frac{y\cos(\theta-\pi/4)+\sqrt{1-y^2}\cos(\theta+\pi/4)}{\sqrt{y^2-1}}e^{-zy}\nonumber\\
& \ = \sqrt{2}\cos(\theta-\pi/4)\int_1^\infty dy\frac{y}{\sqrt{y^2-1}}e^{-zy}\nonumber\\
& \ = \frac{r_1+r_2}{z}\int_1^\infty dy \ z \ e^{-zy}\sqrt{y^2-1} =\frac{r_1+r_2}{z}K_1(z).\nonumber
\end{align}
In the second last equality we used integration by parts.

Finally, consider  the right-hand side of \eqref{temp59022} as
\begin{align}
\sqrt{2}\mcR \int_1^\infty idue^{-z(\cos(\theta-\pi/4)u-i\cos(\theta+\pi/4)\sqrt{u^2-1})}.
\end{align}
Use the substitution $u\rarrow y$, \eqref{temp0609} and then a similar argument to \eqref{temp420999} shows that \eqref{temp59022} holds.
\end{proof}
\end{lemma}

\begin{proposition}\label{prop64}
Let $(r_1,r_2)\in \R^2_{\geq 0}\setminus\{0\}$. The following identities hold
\begin{align}
&2\mcR\int_0^\infty \frac{dt}{\sqrt{1+16t^2}}\exp\big(-\frac{1}{\sqrt{2}}\big(r_1\sqrt{1+4it}+r_2\sqrt{1-4it}\big)\big)\label{temp302}\\& \ =K_0( \ \sqrt{r_1^2+r_2^2} \ ),\nonumber\\
&2\mcR\int_0^\infty \frac{dt}{\sqrt{1+16t^2}}\exp\big(-\frac{1}{\sqrt{2}}\big(r_1\sqrt{1+4it}+r_2\sqrt{1-4it}\big)\big)\big(\frac{1}{\sqrt{2}}\sqrt{1+4it}+\frac{1}{\sqrt{2}}\sqrt{1-4it}\big)\label{temp0596}\\& \ =\frac{r_1+r_2}{\sqrt{r_1^2+r_2^2}}K_1( \ \sqrt{r_1^2+r_2^2} \ ),\nonumber\\
&2\mcR\int_0^\infty \frac{dt}{\sqrt{1+16t^2}}\exp\big(-\frac{1}{\sqrt{2}}\big(r_1\sqrt{1+4it}+r_2\sqrt{1-4it}\big)\big)\big(\frac{1}{\sqrt{2}}\sqrt{1+4it}-\frac{1}{\sqrt{2}}\sqrt{1-4it}\big)\label{temp05961}\\ & \ =\frac{r_1-r_2}{\sqrt{r_1^2+r_2^2}}K_1( \ \sqrt{r_1^2+r_2^2} \ ).\nonumber
\end{align}
\begin{proof}
Changing variables $t\rarrow t/4$ in \eqref{temp302}, one can calculate
\begin{align}
\mcR\sqrt{1+it}=(\frac{1}{2}(\sqrt{1+t^2}+1))^{1/2}=\mcR\sqrt{1-it},\\
\mcI\sqrt{1+it}=(\frac{1}{2}(\sqrt{1+t^2}-1))^{1/2}=-\mcI\sqrt{1-it}
\end{align}
for $t>0$. The exponent in the integrand in \eqref{temp302} is then 
\begin{align}
&-\frac{1}{\sqrt{2}}\big(r_1\sqrt{1+it}+r_2\sqrt{1-it}\big)\\
& \ =-\frac{1}{\sqrt{2}}(r_1+r_2)(\frac{1}{2}(\sqrt{1+t^2}+1))^{1/2}-i\frac{1}{\sqrt{2}}(r_1-r_2)(\frac{1}{2}(\sqrt{1+t^2}-1))^{1/2}.\nonumber
\end{align}
Let $u=(\frac{1}{2}(\sqrt{1+t^2}+1))^{1/2}$, one can calculate
\begin{align}
&t=2u\sqrt{u^2-1}, && \frac{dt}{du}=2\frac{2u^2-1}{\sqrt{u^2-1}},
\\  &(\frac{1}{2}(\sqrt{1+t^2}-1))^{1/2}=\sqrt{u^2-1},&& \sqrt{t^2+1}=2u^2-1.\nonumber
\end{align}
We make the substitution $t\rarrow u$ and get \eqref{temp302} equal to
\begin{align}
&\mcR\int_1^\infty \frac{du}{\sqrt{u^2-1}}\exp\big(-\frac{1}{\sqrt{2}}(r_1+r_2)u-i\frac{1}{\sqrt{2}}(r_1-r_2)\sqrt{u^2-1}\big)\\& \ =K_0( \ \sqrt{r_1^2+r_2^2} \ )\nonumber
\end{align}
by \eqref{temp5290}.

Rescale $t\rarrow t/4$ in both \eqref{temp0596},\eqref{temp05961}, then observe that 
\begin{align}
&\big(\frac{1}{\sqrt{2}}\sqrt{1+it}+\frac{1}{\sqrt{2}}\sqrt{1-it}\big)=\sqrt{2}u,\\
&\big(\frac{1}{\sqrt{2}}\sqrt{1+it}-\frac{1}{\sqrt{2}}\sqrt{1-it}\big)=\sqrt{2}i\sqrt{u^2-1}
\end{align} under the substitution $t\rarrow u$. Hence under the substitution $t\rarrow u$, \eqref{temp0596} becomes
\begin{align}
&\sqrt{2}\mcR\int_1^\infty \frac{udu}{\sqrt{u^2-1}}\exp\big(-\frac{1}{\sqrt{2}}(r_1+r_2)u-i\frac{1}{\sqrt{2}}(r_1-r_2)\sqrt{u^2-1}\big)\\& \ =\frac{r_1+r_2}{\sqrt{r_1^2+r_2^2}}K_1(\sqrt{r_1^2+r_2^2})\nonumber
\end{align}
and \eqref{temp05961} becomes
\begin{align}
&\sqrt{2}\mcR\int_1^\infty i du\exp\big(-\frac{1}{\sqrt{2}}(r_1+r_2)u-i\frac{1}{\sqrt{2}}(r_1-r_2)\sqrt{u^2-1}\big)\\& \ =\frac{r_1-r_2}{\sqrt{r_1^2+r_2^2}}K_1(\sqrt{r_1^2+r_2^2})\nonumber
\end{align}
by Lemma \ref{lemmaintreps}.
\end{proof}
\end{proposition}

\newpage


\end{document}